\documentclass[submission]{dmtcs-episciences}
\usepackage{amsmath}
\usepackage{amssymb,epsfig}
\usepackage[utf8]{inputenc}
\usepackage{a4}

\usepackage{graphicx}
\usepackage{subfig}

\newtheorem{theorem}{Theorem}[section]
\newtheorem{lemma}[theorem]{Lemma}

\newtheorem{corollary}[theorem]{Corollary}

\newtheorem{remark}[theorem]{Remark}
\newtheorem{example}[theorem]{Example}
\newtheorem{definition}[theorem]{Definition}

\def\Z{\mathbb Z}
\def\N{\mathbb N}
\def\R{\mathbb R}
\def\C{\mathbb C}

\def\CC{\mathcal C}

\def\S{\mathcal S}
\def\RR{\mathcal R}

\def\Select{\mathrm{Select}}
\def\SelectM{\mathrm{Select_M}}
\def\SelectD{\mathrm{Select_D}}
\def\Trunc{\mathrm{Trunc}}
\def\Digit{\mathrm{Digit}}
\def\new{\mathrm{new}}
\newcommand{\decdot}{\raisebox{0.1ex}{\textbf{.}}}

\author{
Christiane Frougny\affiliationmark{1} \and
Marta Pavelka\affiliationmark{2} \and
Edita Pelantov\'a\affiliationmark{2} \and
Milena Svobodov\'a\affiliationmark{2}}

\title{On-line algorithms for multiplication and division
in real and complex numeration systems}

\affiliation{IRIF, CNRS and Universit\'e Paris-Diderot
\\
Department of Mathematics, FNSPE, Czech Technical University in Prague}

\received{2018-2-24}
\revised{2019-1-28}
\accepted{2019-4-12}

\begin{document}

\publicationdetails{21}{2019}{3}{14}{4313}
\maketitle

\begin{abstract}

A positional numeration system is given by a base and by a set of digits. The base is a real or complex number $\beta$ such that $|\beta|>1$, and the digit set $\mathcal A$ is a finite set of real or complex digits (including $0$). Thus a number can be seen as a finite or infinite string of digits. An on-line algorithm processes the input piece-by-piece in a serial fashion. On-line arithmetic, introduced by Trivedi and Ercegovac, is a mode of computation where operands and results flow through arithmetic units in a digit serial manner, starting with the most significant digit.

In this paper, we first formulate a generalized version of the on-line algorithms for multiplication and division of Trivedi and Ercegovac for the cases that $\beta$ is any real or complex number, and digits are real or complex.
We then define the so-called OL Property, and show that if $(\beta, \mathcal A)$ has the OL Property, then on-line multiplication and division are realizable by the Trivedi-Ercegovac algorithms. For a real base $\beta$ and
a digit set $\mathcal A$ of contiguous integers, the system $(\beta, \mathcal A)$ has the OL Property if $\# \mathcal A > |\beta|$. For a complex base $\beta$ and symmetric digit set $\mathcal A$ of contiguous integers, the system  $(\beta, \mathcal A)$ has the OL Property if $\# \mathcal A > \beta\overline{\beta} + |\beta + \overline{\beta}|$.
Provided that addition and subtraction are realizable in parallel in the system $(\beta, \mathcal A)$ and that preprocessing of the divisor is possible,  our on-line algorithms for multiplication and division have linear time complexity.

Three examples are presented in detail: base $\beta=\frac{3+\sqrt{5}}{2}$ with digit set $\mathcal A=\{-1,0,1\}$; base $\beta=2i$ with digit set $\mathcal A = \{-2,-1, 0,1,2\}$ (redundant Knuth numeration system); and base $\beta = -\frac{3}{2} + i \frac{\sqrt{3}}{2} = -1 + \omega$, where $\omega = \exp{\frac{2i\pi}{3}}$, with digit set $\mathcal A = \{0, \pm 1, \pm \omega, \pm \omega^2 \}$ (redundant Eisenstein numeration system).
\end{abstract}


\section{Introduction}

A positional numeration system is given by a base and by a set of digits. The base is a real or complex number $\beta$ such that $|\beta|>1$, and the digit set (or alphabet) ${\mathcal{A}}$ is a finite set of real or complex digits (including $0$). The most studied numeration systems are of course the usual ones, where the base is a positive integer. But there have been also numerous studies where the base is an irrational real number (the so-called \emph{$\beta$-expansions}), a complex number, or a non-integer rational number, \textit{etc}. A survey can be found in~\cite[Chapter 2]{cant}. In that setting a number is seen as
a finite or infinite string of digits.

An on-line algorithm processes the input piece-by-piece in a serial fashion, i.e., in the order that the input is given to the algorithm, and the output is produced by the algorithm without having the entire input available from the beginning.

On-line arithmetic, introduced in~\cite{TrivediErcegovac}, is a mode of computation where operands and results flow through arithmetic units in a digit serial manner, starting with the most significant digit. To generate the first digit of the result, the first $\delta$ digits of the operands are required. The integer $\delta$ is called the delay of the algorithm. This technique allows for pipelining of different operations, such as addition, multiplication and division. It is also appropriate for the processing of real (or complex) numbers having infinite expansions: it is well known that when multiplying two real (or complex) numbers, only the left part of the result is significant. On-line arithmetic is used for special circuits such as in signal processing, and for very long precision arithmetic. An application to real-time control can be found in~\cite{DTHL}. One of the benefits of on-line computable functions is that they are continuous for the usual discrete topology on the set of infinite sequences on a finite alphabet. In order to be able to perform on-line computations, it is necessary to use a redundant numeration system, where a number may have more than one representation. A sufficient level of redundancy can also enable parallel addition and subtraction, which are used internally within the on-line multiplication and division algorithms. On the other hand, zero can have in some redundant number systems a non-trivial representation.  This fact requires to modify a representation of a divisor into a suitable form, usually called preprocessing of divisor.

On-line algorithms for multiplication and division in positive integer bases with a symmetric alphabet of integer digits
have been originally given by Trivedi and Ercegovac in~\cite{TrivediErcegovac}.
On-line algorithms for multiplication and division
in some complex numeration systems can be found in~\cite{NielsenMuller}, \cite{McIlhennyErcegovac}, and \cite{FrougnySurarerks}. In this paper we first formulate a generalized version of the on-line algorithms for multiplication and division of Trivedi and Ercegovac for the cases that $\beta$ is any real or complex number, and digits are real or complex.

 Let us say that a pair $(\beta, {\mathcal{A}})$ has the OL Property if there exists  a bounded set $I$  such that
the closure ${\rm cl}(I)$ of $I$  and the interior  $I^o$  of $I$   satisfy $\beta\, {\rm cl}(I) \subset \cup_{a\in {\mathcal{A}}}(I^o +a)$.
 We show that if $(\beta, {\mathcal{A}})$ has the OL Property and $0$ is in $I$ then on-line multiplication and division are realizable by the Trivedi-Ercegovac algorithms. Of course, the divisor inputting the Trivedi-Ercegovac algorithm for division has to be preprocessed into a suitable form.

We show that for a real base $\beta$ and an alphabet ${\mathcal{A}}$ of contiguous integers, the system $(\beta, {\mathcal{A}})$ has the OL Property if $\# {\mathcal{A}} > |\beta|$. For a complex base $\beta$ and a symmetric alphabet ${\mathcal{A}}$ of contiguous integers (convenient for parallel addition and subtraction), the system $(\beta, {\mathcal{A}})$ has the OL Property if $\# {\mathcal{A}} > \beta\overline{\beta} + |\beta + \overline{\beta}|$.

The key point of our algorithms is the specific choice of the functions $\Select$ performing the selection of the digits to output. The definitions of $\Select$ use only a reasonable approximation of its operands by a limited number of fractional digits --- here denoted by $L$ ---  of their $(\beta, {\mathcal{A}})$-representations. This allows, for some specific numeration systems, to perform evaluation of $\Select$ in constant time. In particular, we do not have to treat the real and the imaginary components separately in complex numeration systems.

Provided that addition and subtraction are realizable in parallel in the system $(\beta, {\mathcal{A}})$ (see \cite{FrPeSv1} for general results on this topic) and that preprocessing of the divisor is possible,  our on-line algorithms for multiplication and division have linear time complexity.

Three examples are presented in full detail:\\
1. $\beta=\frac{3+\sqrt{5}}{2}$ and ${\mathcal{A}}=\{-1,0,1\}$: on-line multiplication is possible with delay $\delta = 4$ and with $L= 3$, on-line division with delay $\delta = 6$ and with $L= 9$.\\
2. $\beta=2i$ and ${\mathcal{A}} = \{-2,-1, 0,1,2\}$ (redundant Knuth numeration system): on-line multiplication is possible with delay $\delta = 9$ and $L= 7$,
and on-line division with delay $\delta = 11$ and $L= 11$.\\
3. $\beta = -\frac{3}{2} + i \frac{\sqrt{3}}{2} = -1 + \omega$, where $\omega = \exp{\frac{2i\pi}{3}}$ is the third root of unity, and ${\mathcal{A}} = \{0, \pm 1, \pm \omega, \pm \omega^2 \}$ (redundant Eisenstein numeration system). Here we see that the parameters used in the algorithms are closely linked together and they are not uniquely determined. We present two couples of parameters for the multiplication algorithm: $(\delta, L) = (5, 7)$ and $(\delta, L) = (6, 6)$; and similarly for the division algorithm: $(\delta, L) =  (7, 10)$ and  $(\delta, L) = (10, 9)$.

A short preliminary version of this work has been presented in~\cite{BrFrPeSv}.


\section{Algorithms of Trivedi and Ercegovac}

The on-line multiplication and the on-line division algorithms we describe below are the same as the algorithms introduced by Trivedi and Ercegovac for computation in integer bases with a symmetric alphabet \cite{TrivediErcegovac, ErcegovacLang}. Our modification for non-standard numeration systems for arbitrary base $\beta$ (in general a complex number) and a alphabet ${\mathcal{A}}$ (in general a finite set of complex numbers) concerns only a  choice of the function $\Select$.

In the sequel, by a $(\beta, {\mathcal{A}})$-representation of a number $X$ we understand a (possibly infinite) string  $x_{-n}x_{-n+1}\cdots$ of digits $x_j \in {\mathcal{A}}$ such that  $X =
\sum_{j=-n}^{+\infty} x_j\beta^{-j}$; we also denote it $X =x_{-n}x_{-n+1}\cdots x_0 \decdot x_1x_2 \cdots$.


\subsection{On-line multiplication algorithm}

The algorithm for on-line multiplication in a numeration system $(\beta,{\mathcal{A}})$ has one parameter, namely the delay $\delta\in \N$,  $\delta \geq 1$,   which is specified later. The $\Select$ function is here called $\SelectM$, and has just one variable.

We work with $(\beta,{\mathcal{A}})$-representations of the numbers $X=\sum_{j=1}^\infty x_j\beta^{-j}$ and $Y=\sum_{j=1}^\infty y_j\beta^{-j}$, and their product $P=\sum_{j=1}^\infty p_j\beta^{-j}$. Their partial sums are denoted by $X_k =\sum_{j=1}^kx_j\beta^{-j}$, $Y_k=\sum_{j=1}^ky_j\beta^{-j}$, and $P_k = \sum_{j=1}^k p_j \beta^{-j}$.

The inputs of the algorithm are two (possibly infinite) strings
$$
0\decdot x_1x_2 \cdots x_\delta x_{\delta +1}x_{\delta +2}\cdots \quad \text{ with } x_j \in {\mathcal{A}} \text{ and } x_1 = x_2 = \cdots = x_\delta = 0 , \text{and}
$$
$$
0\decdot y_1y_2 \cdots y_\delta y_{\delta +1}y_{\delta +2}\cdots \quad \text{ with } y_j \in {\mathcal{A}} \text{ and } y_1 = y_2 = \cdots = y_\delta = 0 \, . \text{\ \ \ }
$$

The output is a (possibly infinite) string $ 0\decdot p_1p_2  p_{3}\cdots$ corresponding to a $(\beta,{\mathcal{A}})$-represen\-tation of the product $P = X\cdot Y = \sum_{j=1}^{\infty}p_j\beta^{-j}$. The settings of the algorithm ensure that the representation of $P$ indeed starts only on the right of the fractional point.

We carry out the on-line multiplication in iterative steps. To start with, set $ W_0 = X_0 = Y_0 = p_0=0 \, .$ At the $k$-th step of the iteration (starting from $k=1$) we compute:
\begin{equation} \label{eq:N6}
    W_k = \beta ( W_{k-1} - p_{k-1} ) + ( x_k Y_{k-1} + y_k X_k ) \ , \quad \text{and} \quad p_k = \SelectM(W_{k}) \in {\mathcal{A}} \, .
\end{equation}

\begin{lemma}\label{bounMult}
The definition \eqref{eq:N6} of $W_k$ and $p_k$ implies that, for any $k \geq 1$:
\begin{equation} \label{eq:W_k_mult}
    W_k = \beta^k(X_kY_k-P_{k-1}) \, .
\end{equation}
Moreover, if the sequence $(W_k)$ is bounded, then
$$ X\cdot Y=\lim_{k \rightarrow \infty} X_kY_k = \lim_{k \rightarrow \infty} P_k = P \, . $$
\end{lemma}

\begin{proof}
Due to our setting $x_1 = X_1= y_1=Y_1 = p_0 = P_0 =W_0 =0$, we have by (\ref{eq:N6}) that $W_1 = 0$, and thus $W_1= \beta(X_1Y_1 - P_{0})$. Using again   (\ref{eq:N6}) and the induction hypothesis, we obtain $W_{k+1} = \beta(W_{k} - p_{k}) + x_{k+1}Y_{k} + y_{k+1}X_{k+1} = \beta (\beta^k ( X_k Y_k - P_{k-1}) - p_{k} ) + ( x_{k+1} Y_{k} + y_{k+1} X_{k+1} )$, and the result follows from the fact that $X_{k+1} = X_{k} + x_{k+1} \beta^{-k-1}$, and similar relations for $Y_{k+1}$ and $P_k$.

Thus $X_k Y_k = \beta^{-k} W_k + P_{k-1} = \beta^{-k} (W_k - p_k) + P_k$, and $X_k Y_k - P_k = \beta^{-k} (W_k - p_k)$. Since $p_k$ is from the (finite set) ${\mathcal{A}}$ and $(W_k)$ is bounded, $\lim_{k \rightarrow \infty} X_k Y_k - P_k = 0$.
\end{proof}

The algorithm of Trivedi and Ercegovac gives the following parameters in integer base with $\SelectM(W_k)=\mathrm{round}(W_k)$.

\begin{corollary}\cite{TrivediErcegovac,FrougnySurarerks}
If $\beta$ is an integer $>1$ and ${\mathcal{A}}=\{-a, \ldots, a\}$ with $\beta/2 \le a \le \beta-1$, the on-line multiplication algorithm works with delay $\delta$, where $\delta$ is the smallest positive integer such that
$$
\frac{\beta}{2} + \frac{2a^2}{\beta^{\delta}(\beta-1)} \le a + \frac{1}{2}\,.
$$
\end{corollary}


\subsection{On-line division algorithm}

The algorithm for on-line division in $(\beta,{\mathcal{A}})$ numeration system has two parameters: the delay $\delta\in \N$ and  $D_{\min} >0$,  the minimal value (in modulus) of the denominator. The $\Select$ function is here called $\SelectD$, and it has two variables.

The input consists of $(\beta,{\mathcal{A}})$-representations of the numerator  $N = \sum_{j=1}^\infty n_j \beta^{-j}$, the denominator $D = \sum_{j=1}^\infty d_j \beta^{-j}$, and their quotient $Q = \sum_{j=1}^\infty q_j \beta^{-j}$. Partial sums are denoted by $N_k =\sum_{j=1}^k n_j \beta^{-j}$, $D_k = \sum_{j=1}^k d_j \beta^{-j}$, and $Q_k = \sum_{j=1}^k q_j \beta^{-j}$.

The inputs of the algorithm are two (possibly infinite) strings
$$
0\decdot n_1 n_2 \cdots n_\delta n_{\delta +1} n_{\delta +2} \cdots \quad \text{ with } n_j \in {\mathcal{A}} \text{ and } n_1 = n_2= \cdots = n_\delta = 0 \, , \text{ and}
$$
\begin{equation}\label{Denominator}
    0\decdot d_1 d_2 d_3 \cdots \quad \text{ with } d_j \in {\mathcal{A}} \text{ satisfying } |D_k| \geq D_{\min} \text{ for all } k \in \N, \, k\geq 1 \, .
\end{equation}

The output is a (possibly infinite) string $ 0\decdot q_1 q_2 q_3 \cdots$ corresponding to a $(\beta,{\mathcal{A}})$-represen\-tation of the quotient $Q = N/D = \sum_{j=1}^{\infty} q_j \beta^{-j}$. Again, the settings of the algorithm ensure that the representation of $Q$ starts behind the fractional point.

We carry out the on-line division in iterative steps. To start with, set $W_0 = q_0 = Q_0 = 0 \, .$

Each $k$-th step of the iteration proceeds (starting from $k=1$) by calculation of
\begin{equation}\label{W_k_definition}
    W_{k} = \beta (W_{k-1} - q_{k-1} D_{k-1+\delta}) + (n_{k+\delta} - Q_{k-1} d_{k+\delta}) \beta^{-\delta}.
\end{equation}

The $k$-th digit $q_k$ of the representation of the quotient $Q = \frac{N}{D}$ is evaluated by $\SelectD$, function of the values of the auxiliary variable
$W_k$ and the interim representation $D_{k+\delta}$, so that
\begin{equation} \label{eq:W_k_div}
    q_k=\SelectD (W_k, D_{k+\delta}) \in {\mathcal{A}} \, .
\end{equation}

\begin{lemma}\label{fporadku}
Definition~\eqref{W_k_definition} of $W_k$  implies that, for any $k \geq 1$:
\begin{equation}\label{W_k_formula}
    W_k = \beta^k (N_{k+\delta} - Q_{k-1} D_{k+\delta}) \, .
\end{equation}
Moreover, if the sequence $(W_k)$ is bounded, then $ Q=\lim_{k\to \infty}{Q_k} = \frac{N}{D} \, . $
\end{lemma}

\begin{proof}
Formula (\ref{W_k_formula}) is proved by induction, analogously as in Lemma \ref{bounMult}.
The formula $W_k\beta^{-k} = N_{k+\delta} - Q_{k-1} D_{k+\delta}$ ensures, for bounded $(W_k)$, that
$$
0 = \lim\limits_{k\to \infty} (N_{k+\delta} - Q_{k-1} D_{k+\delta}) = N - D \lim_{k\to \infty}{Q_k} \, .
$$
As $\lim\limits_{k\to \infty}|D_{k+\delta}| = |D| > 0$, the statement follows.
\end{proof}

\medskip

Clearly, the choice of the selection function is the crucial point for correctness of the algorithms for both on-line multiplication and on-line division.


\section{On-line multiplication and division in real and complex bases}\label{S_OL_mult+div}

In this section, we give a sufficient condition on $\beta \in \C$ and ${\mathcal{A}} \subset \C$, which guarantees that the numeration system  $(\beta, {\mathcal{A}})$ allows to perform on-line multiplication and division by the Trivedi-Ercegovac algorithm.

Let us fix the following notation:  for $\varepsilon >0$ and a set $T\subset \C$, $T^\varepsilon$ stands for the $\varepsilon$-fattening of the set $T$:
$$
T^\varepsilon = \bigcup_{x \in T}B(x,\varepsilon)\, , \quad \text{ where } B(x, \varepsilon) \text{ denotes the ball with center } x \text{ and radius } \varepsilon \, .
$$

For numbers $a, \beta \in \C$ and a set $T \subset \C$, we denote
$$
T+a = \{ x +a : x \in T\}\qquad \text{and} \qquad  \beta T = \{ \beta x: x \in T\}\,.
$$
Moreover,  by  ${\rm cl}(T)$   and  $T^o$    we denote the closure and  the interior  of $T$, respectively.
The  metric we use on  $\mathbb{C}$  (or on $\mathbb{R}$,   if    $\beta\in \mathbb{R}$ and  ${\mathcal{A}} \subset  \mathbb{R}$)  is induced     by the absolute  value.

\begin{definition}\label{DefOL}
A pair $(\beta, {\mathcal{A}})$ has the \emph{OL Property} if there exists a bounded set $I$ such that  $\beta\, {\rm cl}(I) \subset \cup_{a\in {\mathcal{A}}}(I^o +a)$.
\end{definition}

The {OL Property}  says that $\beta\, {\rm cl}(I) $ is covered by
copies of $I^o$ shifted by  digits $a$ of $\mathcal{A}$.  Using
Theorem 2.7 in \cite{Kurka}  we can reformulate the OL property.

\begin{lemma} A pair  $(\beta, {\mathcal{A}})$ has OL Property if and only if   there exist a bounded set  $I$  and  a positive number $\varepsilon$ (called the \emph{Lebesgue number} of the covering) such that
\begin{equation}\label{OL}
    \text{for each $x \in (\beta I)^\varepsilon$ there exists $a \in {\mathcal{A}}$ such that $B(x,\varepsilon) \subset I + a$.}
\end{equation}
\end{lemma}

\medskip

The previous consequence  of the  OL Property is crucial for applicability of the Trivedi-Ercegovac algorithm in  non-stnadard numeration systems.

\begin{lemma}\label{digit}
Suppose that $(\beta, {\mathcal{A}})$ has the OL Property, and  let a  bounded set $I$ and $\varepsilon > 0$ satisfy \eqref{OL}. Then there exists a function $\Digit: (\beta I)^\varepsilon \rightarrow {\mathcal{A}}$ such that
\begin{equation}\label{Digit_formula}
 B(V, \varepsilon) \subset I +  \Digit(V)  \,.
\end{equation}
\end{lemma}

When selecting the $k^{th}$-digit $p_k$ (in the multiplication algorithm) or $q_k$ (in the division algorithm), we do not want to evaluate the auxiliary variable $W_k$ precisely, as it would be too costly. We shall use only a reasonable approximation by several most important digits of $W_k$, and also of $D_{k+\delta}$ (for division).

\begin{definition}\label{truncation}
For $E>0$, denote by $\Trunc_E$ a function $\Trunc_E : \C \rightarrow \C$ such that
\begin{equation}
    |X - \Trunc_E(X)| < E \quad \text{ for any } X \in \C.
\end{equation}
\end{definition}

In the sequel, we use the $\Trunc_E(X)$ function in the form of truncation of the less significant digits in the $(\beta, A)$-representation of the number $X = \sum_{j=1}^{\infty} x_j \beta^{-j}$; namely $\Trunc_E(X) = \sum_{j=1}^{L} x_j \beta^{-j}$ with $L \in \N$ such that $| \sum_{j=L+1}^{\infty} x_j \beta^{-j} | < E$.


\subsection{Selection function for on-line multiplication}

Herein, we exploit the OL Property to construct the $\Select$ function for on-line multiplication. According to Lemma \ref{bounMult},  the main and only goal of this construction is to guarantee that the auxiliary sequence $(W_k)$ which is produced by the algorithm remains bounded.

\begin{definition}\label{SelectMult}
Let $(\beta, {\mathcal{A}})$ be a numeration system with the OL Property, let $I \subset \C$ and $\varepsilon > 0$ satisfy \eqref{OL}, and let $\Digit$ be the function from Lemma \ref{digit} satisfying (\ref{Digit_formula}). The \emph{selection function for multiplication} $\SelectM: (\beta I)^{\varepsilon/2} \rightarrow {\mathcal{A}}$ is defined by
\begin{equation}\label{Select_formula}
    \SelectM(U) = \Digit\bigl(\Trunc_{\varepsilon/2}(U)\bigr)  \quad \text{for any} \ U \in ( \beta I)^{\varepsilon/2} \,.
\end{equation}
\end{definition}

The previous definition is correct only if $\Trunc_{\varepsilon/2}(U)$ belongs to the domain of the function $\Digit$. Indeed, since $U \in (\beta I)^{\varepsilon/2}$ and $|U - \Trunc_{\varepsilon/2}(U)| < \varepsilon/2$, the value $\Trunc_{\varepsilon/2}(U)$ is in $(\beta I)^\varepsilon$, as needed.

\begin{lemma}\label{inI}
Let $U \in (\beta I)^{\varepsilon/2}$. Then $U - \SelectM(U) \in I$.
\end{lemma}

\begin{proof}
Let us denote $V = \Trunc_{\varepsilon/2}(U) \in (\beta I)^\varepsilon $ and $a = \Digit(V) \in {\mathcal{A}}$. By the property of the function $\Digit$, we have $B(V,\varepsilon) \subset I+a$. Since $|V-U| <\varepsilon/2$, the value $U \in B(V,\varepsilon) \subset I+a$. Or, equivalently, $U-a \in I$.
\end{proof}

\begin{lemma}\label{AgainIn}
Let $(\beta, {\mathcal{A}})$ be a numeration system with the OL Property, let $I \subset \C$, $\varepsilon > 0$ satisfy \eqref{OL}, and let $\SelectM$ be the function (\ref{Select_formula}) from Definition \ref{SelectMult}. Then there exists $\delta \in \N$ such that, for any $U \in ( \beta I)^{\varepsilon/2}$, any $x,y \in {\mathcal{A}}$, any $X = \sum_{j=\delta+1}^{\infty} x_j\beta^{-j}$ and $Y = \sum_{j=\delta+1}^{\infty} y_j \beta^{-j}$ with $x_j, y_j \in {\mathcal{A}}$, the number
$$
U_{\new} = \beta \bigl(U - \SelectM(U) \bigr) + (y X + x Y ) \ \ \text{belongs to } \ ( \beta I)^{\varepsilon/2}.
$$
\end{lemma}

\begin{proof}
Let us denote $A= \max\{|a| : a \in {\mathcal{A}}\}$, and find $\delta \in \N$ such that
\begin{equation}\label{deltatMult}
    \frac{1}{|\beta|^\delta} \frac{2A^2}{|\beta|-1}< \varepsilon/2\,.
\end{equation}
Then $|yX+xY| < \frac{\varepsilon}{2}$, and, according to Lemma \ref{inI}, the value $\beta \bigl(U - \SelectM(U)\bigr) \in \beta I$. This concludes the proof.
\end{proof}

\begin{theorem}\label{thm:multi}
Suppose that a numeration system $(\beta, {\mathcal{A}})$ has the OL Property, and let a bounded set $I \subset \C$ and $\varepsilon > 0$ satisfy~\eqref{OL}. If $0 \in I$, then on-line multiplication in $(\beta, {\mathcal{A}})$ is performable by the Trivedi-Ercegovac algorithm.
\end{theorem}

\begin{proof}
Since $W_0 = 0 \in I$, necessarily $0\in ( \beta I)^{\varepsilon/2}$. Lemma \ref{AgainIn} implies that $W_k \in( \beta I)^{\varepsilon/2} $ for any $k \in
N$ as well, and thus the sequence $(W_k)$ is bounded. According to Lemma \ref{bounMult}, the boundedness of $(W_k)$ implies that the output sequence $0\decdot p_1p_2p_3\cdots$ converges to the product $P = XY$.
\end{proof}


\subsection{Selection function for on-line division}

Also for on-line division, we need to define the $\Select$ function.  Due to   Lemma \ref{fporadku}, our aim is again to preserve the boundedness of $(W_k)$.

Suppose that the value $D_{\min} >0$ is given, and only divisors satisfying \eqref{Denominator} are on the input of our algorithm. In this whole subsection, we assume that the numeration system $(\beta, {\mathcal{A}})$ has the OL Property, that $I \subset \C$, $\varepsilon > 0$ satisfy \eqref{OL}, and the divisor $D$ satisfies~\eqref{Denominator}.

The $\SelectD$ function in the Trivedi-Ercegovac algorithm for division has two variables, namely $W_k$ and $D_{k+\delta}$. Again, we do not want to compute these values precisely. In order to determine a suitable level of approximation, find $\alpha >0$ such that
\begin{equation}\label{ALFA}
    \alpha \bigl( 1+|\beta|K+ \varepsilon\bigr) < \tfrac\varepsilon2 D_{\min}, \quad \text{where} \ K  = \max\{|x| : x \in I\}\,.
\end{equation}

For specification of the function $\SelectD$ for division, we use the function $\Trunc_\alpha$. Now we moreover require that $\Trunc_\alpha$ fulfils the implication
\begin{equation}\label{trunc_alpha_D}
    |D| > D_{\min} \ \ \Rightarrow \ \ |\Trunc_\alpha(D)| \geq D_{\min}
\end{equation}
for any admissible divisor $D$. This assumption is in fact not restrictive, as we use the $\Trunc$ function in the form of truncation of the less significant digits in the $(\beta, {\mathcal{A}})$-representation of $D$. Since any input $D$ of our algorithm need to satisfy \eqref{Denominator}, the implication \eqref{trunc_alpha_D} is automatically true.

\begin{definition}\label{selectDiv}
Let $U \in \C$ and a divisor $D \in \C$ satisfy $U\in D (\beta I)^{\varepsilon/2}$, and let $\alpha > 0$ fulfil (\ref{ALFA}). The \emph{selection function for division} is defined by
\begin{equation}
    \SelectD(U,D) = \Digit\Bigl(\tfrac{V}{\Delta}\Bigr), \quad \text{where\ }\ V =\Trunc_\alpha(U)\ \text{and} \ \ \Delta = \Trunc_\alpha(D)\,.
\end{equation}
\end{definition}

Let us stress that the domain of the function $\Digit$ is $ (\beta I)^{\varepsilon}$. Thus the previous definition is correct only if $V/\Delta$ belongs to this domain. The next lemma shows that our choice of the parameter $\alpha$ in \eqref{ALFA} guarantees this property.

\begin{lemma}\label{InDomain}
For $U \in \C$ and a divisor $D \in \C$ satisfying \eqref{Denominator}, and for $\alpha > 0$ fulfilling (\ref{ALFA}), put $V = \Trunc_\alpha(U)$ and $\Delta = \Trunc_\alpha(D)$. Then
\begin{equation}
    U \in D (\beta I)^{\varepsilon/2} \quad \Longrightarrow \quad V\in \Delta (\beta I)^{\varepsilon}\,.
\end{equation}
\end{lemma}

\begin{proof}
For $U \in D (\beta I)^{\varepsilon/2}$, there exist $y \in I$ and $\varepsilon_1 \in \C$ such that $U = D (\beta y +\varepsilon_1)$ and $|\varepsilon_1|< \varepsilon/2$. Let us denote $\alpha_1, \alpha_2 \in \C$ such that $U = V + \alpha_1$ and $D = \Delta + \alpha_2$. Obviously, $|\alpha_1|$, $|\alpha_2| <\alpha$. We get $V + \alpha_1 = (\Delta + \alpha_2)(\beta y + \varepsilon_1)$. Thus $V = \Delta (\beta y + \varepsilon_1) - \alpha_1 + \alpha_2 (\beta y + \varepsilon_1)$. If we denote
$$
\varepsilon_2 = \frac{-\alpha_1 + \alpha_2 (\beta y + \varepsilon_1)}{\Delta} \, ,
$$
we can express $V = \Delta (\beta y + \varepsilon_1 + \varepsilon_2)$. Using \eqref{ALFA}, we obtain
$$
|\varepsilon_2| \leq \frac{\alpha +\alpha (|\beta|K + \varepsilon/2 )}{D_{\min}} < \frac{\varepsilon}{2} \,.
$$
It means that $V = \Delta (\beta y + \varepsilon_1 + \varepsilon_2)$ belongs to $\Delta (\beta I)^{\varepsilon}$.
\end{proof}

The following statement corresponds to the iterative step in the division algorithm.

\begin{lemma}\label{inductionDivision}
There exists $\delta \in \N$ such that, for any $U, D, F, G \in \C$ with the properties $U \in D (\beta I)^{\varepsilon/2}$, $|F| \leq  A=\max\{|a| : a \in {\mathcal{A}}\}$ and $|G| \leq A \Bigl(1 + \tfrac{A}{|\beta|-1} \Bigr)$, the numbers $U_{\new} = \beta (U - qD) + \tfrac{G}{\beta^\delta}$ and $D_{\new} = D + \tfrac{F}{\beta^{\delta+1}}$, where $q = \SelectD(U,D)$, satisfy $ U_{\new} \in D_{\new}(\beta I)^{\varepsilon/2} \, .$
\end{lemma}

\begin{proof} For $V=\Trunc_\alpha(U)$ and $\Delta=\Trunc_\alpha(D)$, denote $\alpha_1$ and $\alpha_2$ such that $U=V+\alpha_1$ and $D=\Delta+ \alpha_2$. Clearly, $|\alpha_1|, |\alpha_2| < \alpha$. Let us rewrite
$$
\beta (U - q D) = \beta \Bigl(V + \alpha_1 - q(\Delta + \alpha_2)\Bigr) = \beta \Delta \Bigl(\tfrac{V}{\Delta} - q \Bigr)+\beta \alpha_1 - \beta q \alpha_2 \,.
$$
Using $D_{\new} = \Delta + \alpha_2 + F \beta^{-\delta-1}$, we obtain
$$
\beta (U - q D) = \beta D_{\new} \Bigl(\tfrac{V}{\Delta} - q \Bigr) - \beta \bigl(\alpha_2 + F \beta^{-\delta - 1}\bigr)\Bigl(\tfrac{V}{\Delta} - q \Bigr)+ \beta \alpha_1 - \beta q \alpha_2 \,.
$$
Thus
$$
U_{\new} = \beta D_{\new} \Bigl(\tfrac{V}{\Delta} - q \Bigr) + \beta C_1 D_{\new} + \tfrac{1}{\beta^\delta}C_2 D_{\new} \, ,
$$
where the values $C_1, C_2 \in \C$ are found so that they satisfy $C_1 D_{\new} = \alpha_1 - \alpha_2 \tfrac{V}{\Delta}$ and $C_2 D_{\new} = G - F \Bigl( \tfrac{V} {\Delta} - q \Bigr)$. By Lemma \ref{InDomain}, we know that $\left|\tfrac{V} {\Delta}\right| < |\beta| K + \varepsilon$, and thus the modulus of $C_1$ can be, by virtue of \eqref{ALFA}, bounded by
$$
|C_1| \leq \frac{\alpha + \alpha( K|\beta|+\varepsilon)}{|D_{\new}|} \leq \frac{\alpha (1+|\beta|K + \varepsilon)}{D_{\min}} <\tfrac{\varepsilon}{2} < \varepsilon \, .
$$
Thanks to $\frac{V}{\Delta} \in (\beta I)^{\varepsilon}$, the choice of the function $\Digit$ implies that $q = \Digit(\tfrac{V}{\Delta})$ satisfies $B(\tfrac{V}{\Delta}, \varepsilon) \subset I +q$. As $|C_1|<\varepsilon$, we have  $\tfrac{V}{\Delta} + C_1 \in B(\tfrac{V}{\Delta}, \varepsilon)\subset I + q$ . Or, equivalently, $\tfrac{V}{\Delta} - q + C_1 \in I$. We can write
$$
U_{\new} \in D_{\new}(\beta I + \tfrac{C_2}{\beta^\delta}) \, .
$$
To complete the proof, we need to find $\delta$ such that $\Bigl| \tfrac{C_2}{\beta^\delta} \Bigr| < \tfrac{\varepsilon}{2}$. The number $C_2$ can be bounded as follows:
$$
|C_2| \leq \frac{|G| + |F| \Bigl| \tfrac{V}{\Delta} - q \Bigr|}{|D_{\new}|}\leq\frac{1}{D_{\min}} \left( A\Bigl(1 + \tfrac{A}{|\beta| - 1}\Bigr) + A (K + \varepsilon)\right) \, .
$$

Therefore, it is sufficient (and possible at the same time) to choose $\delta \in \N$ such that
\begin{equation}\label{DeltaDivision}
    \frac{A}{D_{\min}} \left( 1+\tfrac{A}{|\beta|-1} + K + \varepsilon \right) < \tfrac{\varepsilon}{2} \ |\beta|^\delta\,.
\end{equation}
\end{proof}

\begin{theorem}\label{thm:division} Let $I \subset  \C$  and  $\varepsilon > 0$ ensure the OL Property of a numeration system $(\beta, {\mathcal{A}})$ and  $0 \in I$. Let  strings $0\decdot n_1 n_2 \cdots $  and  $0\decdot d_1 d_2 \cdots $ satisfying \eqref{Denominator} represent numbers $N$ and $D$ resp. Then computing  $N/D$ can be performed by the Trivedi-Ercegovac algorithm.

\end{theorem}

\begin{proof}
Let $\alpha > 0$ be chosen to fulfill \eqref{ALFA}, and the delay equal to $\delta$ from Lemma \ref{inductionDivision}. For the Trivedi-Ercegovac division algorithm, we use the function $\SelectD$ from Definition \ref{selectDiv}. According to Lemma \ref{fporadku}, for correctness of the algorithm one has to show that the sequence $(W_k)$ is bounded.

We prove by induction on the index $k \in \N$ that, for each $k \geq 0$, the value $W_k$ satisfies $W_{k} \in  D_{\delta +k}( \beta I)^{\varepsilon/2} \, .$

As $W_0=0 \in I$, obviously $W_0\in (\beta I)^{\varepsilon/2}D_{\delta}$. According to \eqref{W_k_definition}, the value $W_{k+1}$ is determined from $W_k$ by
$$
W_{k+1} = \beta (W_k - q_k D_{k+\delta}) + (n_{k+1+\delta} - Q_k d_{k+1+\delta}) \beta^{-\delta} \, ,
$$
and
$$
D_{k+1+\delta} = D_{k + \delta } + \tfrac{d_{k+1+\delta}}{\beta^{\delta + 1 + k}} \, .
$$
Now we apply Lemma \ref{inductionDivision} with $U = W_k$, $D = D_{k+\delta}$, $F = \tfrac{d_{\delta+k+1}}{\beta^k}$ and $G = n_{k+1+\delta} - Q_k d_{k+1+\delta}$, and obtain the implication
$$
W_k \in D_{\delta +k}( \beta I)^{\varepsilon/2} \quad \Longrightarrow \quad W_{k+1} \in D_{\delta +k+1}( \beta I)^{\varepsilon/2} \, .
$$
The OL Property guarantees that the set $I$ is bounded, and the values $D_k$ are bounded by $\tfrac{A}{|\beta|-1}$ in modulus. Thus the sequence $(W_k)$ is bounded too, as we wanted to prove.
\end{proof}


\section{OL Property}

For any (real or complex) base $\beta$, there exists a suitable alphabet ${\mathcal{A}}$ such that $(\beta, {\mathcal{A}})$ has the OL Property. For instance, the set $I$ fulfilling the OL Property can be the ball $B(0, 1) = I$ with the alphabet  ${\mathcal{A}}$ containing a sufficient number of elements.

Nevertheless, note that the alphabet ${\mathcal{A}}$ may generally be any subset of complex numbers containing zero. We have no general method to verify whether a given numeration system $(\beta, {\mathcal{A}})$ has the OL Property and, in particular, we are not able to check the OL Property for the most interesting alphabet, namely the minimal alphabet ${\mathcal{A}}$ allowing parallel addition and subtraction in a given base $\beta$.

We focus our attention on alphabets of contiguous integers containing zero. In the case of complex bases we study only symmetric alphabets. This restriction is in fact quite innocent, since such alphabets are preferable with respect to parallel addition and subtraction. For both real and complex bases, with alphabets of contiguous integers, we provide a straightforward manner for finding the set $I$ and checking the OL Property.


\subsection{OL Property for real bases and integer alphabets}

Redundancy of a numeration system is a necessary condition for any on-line algorithm. In this section, we consider real bases $\beta$ and alphabets ${\mathcal{A}}$ of contiguous integers containing zero. For such a system, redundancy is characterized by the inequality  $\#{\mathcal{A}} > |\beta|$. We will show that redundancy of a real numeration system (with an alphabet of contiguous integers) is also a \emph{sufficient condition} for the Trivedi-Ercegovac algorithm.

\begin{lemma}\label{intervalI}
Let $\beta$ be a real number with $|\beta|>1$ and let ${\mathcal{A}} = \{ m, \ldots, 0, \ldots, M \} \subset \Z$ with $m \leq 0 \leq M$. If $|\beta| < \#{\mathcal{A}} = M - m + 1$, then the numeration system $(\beta, {\mathcal{A}})$ has the OL Property.

\noindent  In particular:
\begin{itemize}
    \item for $\beta >1$, one of the pairs $(I,\varepsilon)$ satisfying \eqref{OL} is $I = [\lambda, \rho]$ and $\varepsilon >0$ defined by
            \begin{equation}\label{real_positive}
                \varepsilon = \frac{M - m + 1 - \beta}{2 (\beta+1)} > 0, \qquad \rho = \frac{M - 2 \varepsilon}{\beta-1}, \qquad \lambda = \frac{m + 2\varepsilon}{\beta-1} \, ;
            \end{equation}
    \item for $\beta < -1$, one of the pairs $(I,\varepsilon)$ satisfying \eqref{OL} is $I = [\lambda, \rho]$ and $\varepsilon > 0$ defined by
            \begin{equation}\label{real_negative}
                \varepsilon = \frac{M - m + 1 + \beta}{2(1-\beta)} > 0, \qquad \rho = \frac{1-m}{1-\beta}, \qquad \lambda = \frac{-M - 1}{1-\beta} \, .
            \end{equation}
\end{itemize}
\end{lemma}

\begin{proof}
Consider $\beta >1$. Since $\rho=\lambda +1 + 2\varepsilon$, the overlap of intervals $(I+a)\cap(I+(a+1)) = a+[\lambda +1, \rho]$ is of length $2\varepsilon$ for any $a\in\Z$. Equations (\ref{real_positive}) imply that $\beta \rho + 2\varepsilon = \rho + M$ and $\beta \lambda - 2\varepsilon = \lambda + m$, thus the $(2\varepsilon)$-fattening of $\beta I$ equals $(\beta I)^{2\varepsilon} = \cup_{a\in {\mathcal{A}}}(I+a)$, and \eqref{OL} holds. In the case of $\beta <-1$, by use of equations (\ref{real_negative}) we also obtain $\rho = \lambda + 1 + 2\varepsilon $. Since $\beta \lambda + 2\varepsilon = \rho + M$ and $\beta \rho - 2\varepsilon = \lambda + m$, the statement \eqref{OL} holds here as well.
\end{proof}

\begin{remark}\label{positive_extension}
If $\beta < -1$, then the interval $I=[\lambda, \rho ]$ in Lemma \ref{intervalI} always contains $0$. The same is true if $\beta >1$ and $m < 0 < M$. Thus, according to Theorems \ref{thm:multi} and \ref{thm:division}, the on-line algorithms work properly.

If $\beta >1$  and $M=0$, i.e., the alphabet consists of non-positive integers, then only non-positive numbers have a $(\beta, {\mathcal{A}})$-representation. Product or quotient of such numbers is positive, and thus without
any $(\beta, {\mathcal{A}})$-representation. Therefore, no (on-line) algorithm for multiplication or division makes sense in this case.

If $\beta > 1$ and $m = 0$, i.e., the alphabet consists of non-negative integers, then no interval $I \subset \R$ suitable for the OL Property contains $0$. Nevertheless, even in this case the Trivedi-Ercegovac algorithm can be used. The $\Select$ function just has to be slightly modified as follows: Consider the interval $I= [\lambda, \rho]$ from Lemma \ref{intervalI}. In particular, the left boundary of the interval is $\lambda= \tfrac{2\varepsilon}{\beta-1} >0$. The $\SelectM$ function given by Definition \ref{SelectMult} has as its domain the interval $(\beta I)^{\varepsilon/2}$ with its left boundary $\beta \lambda - \tfrac{\varepsilon}{2}$. Put
$$
\widetilde{\SelectM} (U) = \left\{ \begin{array}{cl}
    0            & \text{if } \ U < \beta \lambda - \tfrac{\varepsilon}{2} \, , \\
    \SelectM (U) & \text{if } \ U \geq \beta \lambda - \tfrac{\varepsilon}{2} \, .
\end{array}\right.
$$
Using this extended $\widetilde{\SelectM}$ function in the algorithm for multiplication (and analogously also the extended $\widetilde{\SelectD}$ function for division) and starting with $W_0 = 0$, we get the digit $p_0 = 0$ at the beginning on the output. Consequently, as long as $p_k = 0$, it holds that $W_k \geq \beta W_{k-1}$, due to \eqref{eq:N6}. Thus the sequence $(W_k)$ is increasing, and after several iterations, $W_k$ reaches the interval $(\beta I)^{\varepsilon/2}$. According to Lemmas \ref{AgainIn} and \ref{inductionDivision}, the value of $W_k$ then stays in $(\beta I)^{\varepsilon/2}$ in all further steps. Thus the sequence $(W_k)$ is bounded, and the algorithms work properly.
\end{remark}

A direct consequence of Lemma~\ref{intervalI} and Remark \ref{positive_extension} is the following result.

\begin{theorem}\label{realOL}
Let $\beta$ be a real number with $|\beta|>1$ and ${\mathcal{A}} = \{m, \ldots, 0, \ldots, M\} \subset \Z$. Let us assume that $m \leq 0 < M$ for $\beta > 1$, and   $m \leq 0 \leq M$ for $\beta <-1$. If $|\beta| < \#{\mathcal{A}} = M - m + 1$, then multiplication and division in the numeration system $(\beta, {\mathcal{A}})$ are on-line performable by the Trivedi-Ercegovac algorithms.
\end{theorem}


\subsection{OL Property for complex bases and integer alphabets}

The aim of this section is to prove that for any complex base $\beta \in \C$, it is always possible to find a sufficiently large symmetric alphabet ${\mathcal{A}}$ of contiguous integers, so that the system $(\beta, {\mathcal{A}})$ has the OL Property and also allows parallel addition and subtraction. Parallel addition and subtraction are the reason for choosing such a specific form of an alphabet,  see \cite{FrHePeSv} for more details.

The result we present in this section for complex bases is somehow weaker then the one presented in the previous section for real bases.

At first, let us stress two facts about the OL Property, which follow directly from its definition. Supposing that a bounded set $I$ and $\varepsilon>0$ ensure the OL Property for $(\beta, {\mathcal{A}})$, then:
\begin{itemize}
    \item $\overline{I}$ and $\varepsilon>0$ ensure the OL Property for $(\overline{\beta}, \overline{ {\mathcal{A}}})$, where $\overline{z}$ denotes the complex conjugate of the number $z$; and
    \item if $-I = I$, then $I$ and $\varepsilon>0$ ensure the OL Property for $(-\beta, {\mathcal{A}})$.
\end{itemize}

At the end of this section and also in Section~\ref{examples}, we prove the OL Property for two specific numeration systems with complex base and complex alphabet. However, in the case of a  complex base we manage to
provide a general result only for systems with a symmetric alphabet of contiguous integers:

\begin{theorem}\label{complex_beta+integer_A}
Let $\beta \in \C\setminus\R$, $|\beta|>1$ and ${\mathcal{A}} = \{ -M, \ldots, -1, 0, 1, \ldots, M\} \subset \Z $. If
\begin{equation}\label{A_size:complex_base}
    \beta \overline{\beta} + |\beta + \overline{\beta}| < \# {\mathcal{A}} = 2M+1 \, ,
\end{equation}
then the numeration system $(\beta, {\mathcal{A}})$ has the OL Property.
\end{theorem}

\begin{proof}
First we assume that the real and imaginary parts of $\beta$ fulfil $\Re \beta \leq 0$ and $\Im \beta >0$. We define $I \subset \C$ to be a parallelogram with vertices $A, B, -A, -B$. Clearly, $I$ is centrally symmetric (i.e., $-I = I$). We choose the points $A = A_1 + i A_2, B = B_1 +i B_2 \in \C$ to satisfy
\begin{equation}\label{AB}
    0< \Im  A = \Im  B\,, \quad \Re  A < \Re B \quad \text{and} \quad \Im \,(\beta B) = \Im \,(-\beta A) \, .
\end{equation}
The previous assumptions imply $A_2 = B_2$, $0< A_2$ and $B_1 = A_1 +  2x_0$ for some $x_0 > 0$. In this notation, $A + B = 2 (A + x_0)$, and thus the equality  $\Im \,(\beta B) = \Im \,(-\beta A)$ gives $0 = \Im \,\bigl(\beta(A + B)\bigr) = \Im (2 \beta (A + x_0)) = 2 ( (A_1 + x_0) \Im \beta + A_2 \Re \beta)$. It implies
\begin{equation}\label{(-1)}
    A_1 + x_0 = -\frac{\Re \beta}{ \Im \beta} A_2 \, .
\end{equation}
Consequently, if we fix $x_0 > 0$ and $A_2 > 0$, then the points $A, B$ are fully determined by \eqref{AB}. The sets $I$ and $\beta I$ are depicted on Figure~\ref{figure_OL_11}.

%
	\begin{figure}[h]
		\centering
		\subfloat[]{\includegraphics[width=0.45\textwidth]{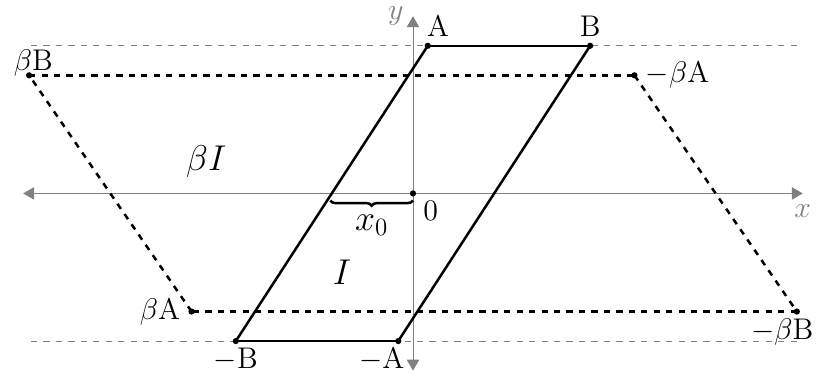}\label{figure_OL_11}}
		\hfill
		\subfloat[]{\includegraphics[width=0.45\textwidth]{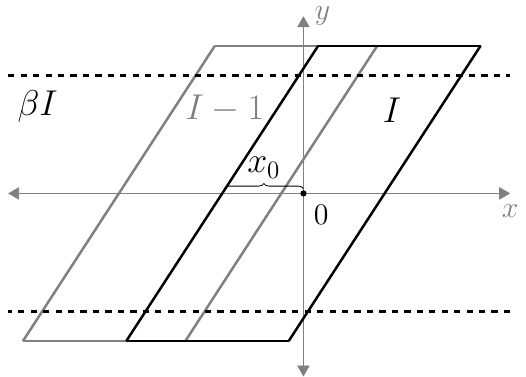}\label{figure_OL_12}}
		\caption{Construction of the set $I$ fulfilling the OL Property for a numeration system with complex base and symmetric integer alphabet.}
        \label{OL_complex}
	\end{figure}


Choose the length of the edge between $A$ and $B$, i.e., the value $2 x_0$ to be bigger than $1$, namely:
\begin{equation}\label{x_0}
    2x_0>1 \,.
\end{equation}
This choice guarantees that the interiors of the neighboring copies of $I$ overlap, i.e., $(I^o +a) \cap (I^o +a+1) \neq \emptyset$ for all $a, a+1 \in {\mathcal{A}}$, see Figure~\ref{figure_OL_12}. Consequently, the set $\cup_{a\in {\mathcal{A}}}(I+a)$ is the parallelogram with vertices $A-M, B+M, -A+M, -B-M$, which is centrally symmetric. If the coordinates of the parallelograms $\beta I$ and $\bigcup_{a\in{\mathcal{A}}} (I+a)$ satisfy the following inequalities:
\begin{equation}\label{(1)}
    \Im A =\Im B > \Im (\beta B) \, ,
\end{equation}
\begin{equation}\label{(2)}
    \Re(\beta B)  > \Re A - M \, ,
\end{equation}
\begin{equation}\label{(3)}
    \Re(\beta A)  > \Re (-B) - M \, ,
\end{equation}
then the set $\beta I$ is covered by interiors  of copies of $I$, i.e.,
$$
\beta I \subset \bigcup_{a\in {\mathcal{A}}}(I^o+a) \, .
$$
To complete the proof, we have to find $x_0 > 0$ and $A_2 > 0$ such that the four  inequalities \eqref{x_0}, \eqref{(1)}, \eqref{(2)} and \eqref{(3)} hold. Let us express the inequalities \eqref{x_0}, \eqref{(1)} and \eqref{(2)} using $x_0$  and $A_2$. (The inequality \eqref{(3)} will be discussed later.)

As $\beta B = \tfrac12 \beta (A+B) + \tfrac12 \beta (B-A)$, by \eqref{AB} we have $\Im (\beta B) =  \tfrac12 \Im \bigl(\beta (A+B) \bigr) + \tfrac12 \Im \bigl(\beta (B-A) \bigr) = 0 + \tfrac12 \Im \bigl(\beta (2 x_0) \bigr) = x_0 \Im \beta$. Thus, the inequality \eqref{(1)} in fact means:
\begin{equation}\label{(4)}
    A_2 > x_0 \Im \beta \,.
\end{equation}
As $B=A+2x_0$, we have
$$
\Re (\beta B) = 2x_0 \Re \beta + A_1\Re \beta - A_2\Im \beta = x_0 \Re \beta + (x_0+A_1)\Re \beta - A_2\Im \beta \, .
$$
Using \eqref{(-1)}, we obtain
$$
\Re (\beta B) = x_0 \Re \beta -\frac{(\Re \beta)^2}{ \Im \beta} A_2- A_2\Im \beta=   x_0 \Re \beta - \frac{A_2}{ \Im \beta} \Bigl(\bigl(\Re \beta\bigr)^2 + \bigl(\Im \beta\bigr)^2\Bigr) \, .
$$
The inequality \eqref{(2)} may thus be reformulated (using \eqref{(-1)} repeatedly) into
\begin{equation}\label{(5)}
    M > \frac{A_2}{ \Im \beta} (\beta \overline{\beta} - \Re \beta) - x_0(\Re \beta +1)\,.
\end{equation}

We now work with $\beta$ satisfying $\Re \beta \leq 0$. For such $\beta$, the assumption \eqref{A_size:complex_base} on cardinality of the alphabet $\# {\mathcal{A}} = 2M+1 > \beta \overline{\beta} + |\beta + \overline{\beta}|$ means:
$$
M > \tfrac12 \bigl( \beta\overline{\beta} - 2 \Re \beta - 1 \bigr) \,.
$$
This strict inequality allows us to find $x_0 > \tfrac12$ such that
$$
M> x_0 \bigl( \beta\overline{\beta} - 2\Re \beta  -1\bigr)= x_0 \bigl( \beta\overline{\beta} - \Re \beta \bigr)  - x_0\bigl( \Re \beta +1\bigr)\,.
$$
Again, because of the previous strict inequality and the assumption $\Im \beta >0$, one can find $A_2>0 $ such that
$$
\frac{A_2}{\Im \beta} > x_0 \quad \text{and} \quad  M > \frac{A_2}{\Im \beta} \bigl(\beta\overline{\beta} - \Re \beta \bigr) - x_0 \bigl(\Re \beta + 1 \bigr)\,.
$$
It means that there exist $x_0$ and $A_2$ such that \eqref{x_0}, \eqref{(4)} and \eqref{(5)} are fulfilled, or, equivalently, \eqref{x_0}, \eqref{(1)} and \eqref{(2)} are fulfilled.

It remains to show that \eqref{(3)} is satisfied as well. We do so by proving that $\Re(\beta A) \geq \Re(\beta B)$ and $\Re A \geq \Re (-B)$, and thus validity of \eqref{(2)} implies validity of \eqref{(3)}.

As $B = A + 2 x_0$, we have $\Re (\beta A) = \Re (\beta B) - 2 x_0 \Re \beta$. Since $x_0 > 0$ and $\Re \beta \leq 0$, obviously $\Re(\beta A) \geq \Re(\beta B)$.

From $A + B = 2 (A + x_0)$, we get by \eqref{(-1)} that $\Re (A + B) = 2 (A_1 + x_0) > 0$, and thus $\Re A \geq \Re (-B)$.

Now we can summarize that the proof of the theorem is complete for the case $\Re \beta \leq  0$ and $\Im \beta > 0$. Since the set $I$ we used to demonstrate the OL Property is centrally symmetric (i.e., $-I = I$) and the alphabet satisfies $\overline{{\mathcal{A}}} = {\mathcal{A}}$, the OL Property is possessed also by the numeration systems $(-\beta, {\mathcal{A}})$, $(\overline{\beta}, {\mathcal{A}})$ and $(-\overline{\beta}, {\mathcal{A}})$. Therefore, the proof is complete for all bases $\beta$, $|\beta| > 1$ with $\Im \beta \neq 0$.
\end{proof}

\begin{remark}
In case  $\Re \beta =0$, the condition of Theorem \ref{complex_beta+integer_A} has the form $\# {\mathcal{A}} > \beta\overline{\beta}$. For a complex numeration system, this is in fact the necessary condition of redundancy as defined in \cite{Kurka}.   Therefore, the bound given in Theorem \ref{complex_beta+integer_A} is optimal for the case $\Re \beta = 0$. An example of such a numeration system is the Knuth system with base $\beta =2i$, for details see Section \ref{examples}. If $\Re \beta \neq 0$, we do not know whether the bound $\beta \overline{\beta} + |\beta + \overline{\beta}|$ is optimal for $\# {\mathcal{A}}$. Unlike for real bases, we have no general result for complex bases with alphabets of contiguous integers ${\mathcal{A}}\subset\Z$ which are not centrally symmetric.
\end{remark}

\medskip

For a  complex base $\beta$ a  complex alphabet may be preferable. For instance,
the alphabet   ${\mathcal{A}} = \{0, \pm 1, \pm i\}$   is closed under multiplication and allows parallel addition  with the base $\beta = -1+i$ (the  redundant Penney numeration system). Fig.~\ref{Penney_picture} shows that this numeration system has the (OL) Property.

\begin{figure}[h]
    \centering
    \includegraphics[scale=0.45]{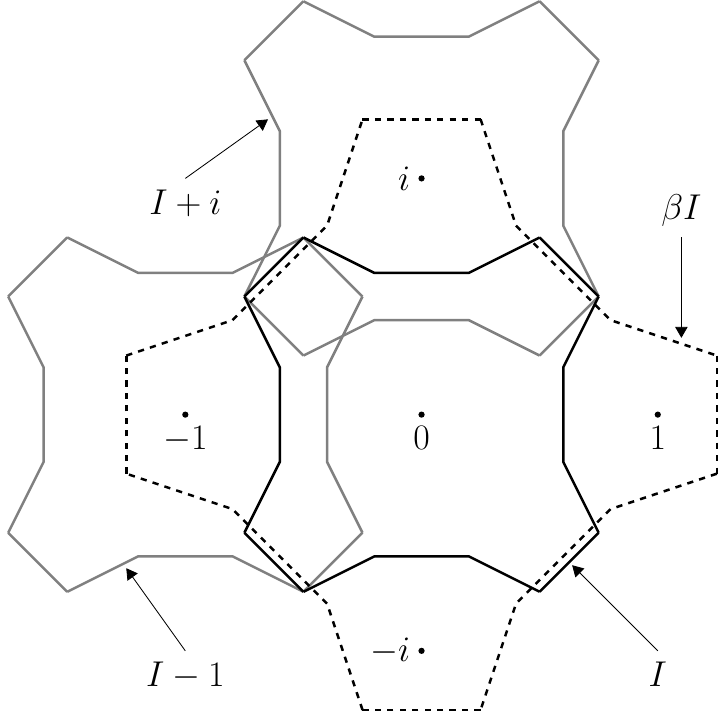}
    \caption{Penney numeration system with base $\beta = -1 + i$ and alphabet ${\mathcal{A}} = \{0, \pm 1, \pm i\}$ fulfills the OL Property, due to the ``star-shaped" set $I$ illustrated hereby.
    \label{Penney_picture}}
\end{figure}



\section{Parameters in on-line algorithms}

In this whole section, we assume that the numeration system $(\beta, {\mathcal{A}})$ satisfies the OL Property. In order to be able to use the on-line algorithms, we need to determine one parameter, namely $\delta$, for multiplication, and two parameters $\delta$ and $D_{\min}$ for division. The inequalities \eqref{deltatMult} and \eqref{DeltaDivision} provide formulae for $\delta$, given the bounded set $I \subset \C$ and the parameter $\varepsilon > 0$ from the OL Property, and also given the parameter $D_{\min}$.
The main attention in this  section is devoted to  the problem of how to determine  $D_{\min}$.  At the end of this section we touch the  question of the optimality of the parameters occurring in the on-line algorithms.


\subsection{Preprocessing of divisor and $D_{\min}$}

By {\it preprocessing of divisor}, we mean a transformation of the divisor into the form required in \eqref{Denominator}:
$$
0 \decdot d_1 d_2 d_3 \cdots \quad \text{ with } d_j \in {\mathcal{A}} \text{ satisfying } |D_k| = \bigl| \sum_{j=1}^k d_j \beta^{-j} \bigr| \geq D_{\min} \ \text{for all } k \in \N, k \geq 1 \, .
$$


\medskip

In particular, for $k=1$, we need $d_1 \neq 0$. Therefore, the transformation consists at least in shifting the fractional point to the most significant non-zero digit of the representation of the divisor, i.e., we multiply the divisor by a suitable power of $\beta$, and, after obtaining the result of the division, we must take this fact into account.

Let us denote
\begin{equation}\label{defR}
    \RR = \Bigl\{\bigl|  \sum_{j=1}^{\infty}d_j\beta^{-j} \bigr|\  :\   d_1\neq 0, d_j \in {\mathcal{A}}\Bigr\}.
\end{equation}
If $\inf \RR > 0$, then one can put $D_{\min} = \inf \RR$ into the on-line algorithm for division, and nothing else than shifting the fractional point is needed.  In our further considerations about the parameter $D_{\min}$, the following notion plays a key role.

\begin{definition}
Let $(\beta, {\mathcal{A}})$ be a numeration system. If $0= \sum_{j=1}^{\infty} z_j\beta^{-j} $, where $z_j \in {\mathcal{A}}$ for all $ j \geq 1$ and $z_k\neq 0$ for at least one index $k$, then  the sequence $z_1, z_2, z_3, \ldots$ is called a \em{non-trivial $(\beta, {\mathcal{A}})$-representation of zero}.
\end{definition}

The relation between representations of zero and $\RR$ is obvious:

\begin{lemma}\label{R}
$\inf \RR = 0$ if and only if $0$ has a non-trivial $(\beta, {\mathcal{A}})$-representation.
\end{lemma}

As already mentioned, in numeration systems without non-trivial representations of zero, the determination of $D_{\min }$ and the preprocessing of the divisor are simple. In particular,  if   $\beta$  is positive and $\mathcal{A}$ contains only non-negative or only non-positive digits, then  $0$ has only the trivial representation  and we can take  $D_{\min} = \tfrac{1}{\beta} \min\{|a|: a\in \mathcal{A}\}$. In numeration systems having a non-trivial representation of zero, the determination of $D_{\min }$ and the divisor preprocessing are more laborious, and no general recipe applicable to all bases is available. The following lemma helps to identify such numeration systems.

\begin{lemma}\label{ReprezZero}
Let $\beta > 1$ and $\{ -1, 0, 1 \} \subset {\mathcal{A}} = \{m, \ldots, 0, \ldots, M\} \subset \Z$. Then $0$ has a non-trivial $(\beta, {\mathcal{A}})$-representation if and only if
$$
\beta \leq \max \{M + 1, -m + 1\} \, .
$$
\end{lemma}

\begin{proof}
Let $z =  \sum_{j=1}^{\infty} z_j \beta^{-j}$ with $z_j \in {\mathcal{A}}$ and $z_1\neq 0$.
\begin{itemize}
    \item If $z_1\geq 1$, then $z \geq \tfrac1{\beta} + m  \sum_{j=2}^{\infty}\beta^{-j} = \tfrac{\beta - 1 + m}{\beta(\beta - 1)}$.
    \item If $z_1\leq -1$, then $z \leq -\tfrac1{\beta} + M  \sum_{j=2}^{\infty}\beta^{-j} = \tfrac{-\beta + 1 + M}{\beta(\beta -1)}$.
\end{itemize}
Obviously, if $\beta > \max \{M + 1, -m + 1\}$, then $\inf\RR = \min\{\tfrac{\beta-1-M}{\beta(\beta -1)}, \tfrac{\beta-1+m}{\beta(\beta -1)}\}>0$, and zero has only the trivial $(\beta, {\mathcal{A}})$-representation.

For showing the opposite implication, we use a result of R\'enyi \cite{Renyi}. For any base $\beta >1$, the number $1$ can be written in the form $1 = \sum_{j=1}^\infty t_j \beta^{-j}$, where $t_j \in \{z \in \Z : 0 \leq z < \beta\}$. In particular, it means that
$$
0 = \frac{1}\beta - \sum_{j=1}^\infty t_j\beta^{-j-1} = -\frac{1}{\beta} + \sum_{j=1}^\infty t_j \beta^{-j-1} \, .
$$
If one of the sets $\{z \in \Z  :0 \leq z < \beta\}$ or $\{-z \in \Z : 0\leq z < \beta\}$ is a subset of ${\mathcal{A}}$, then $0$ has a non-trivial $(\beta, {\mathcal{A}})$- representation.
\end{proof}

\begin{example}\label{base4}
If $\beta = 4$ and ${\mathcal{A}} = \{-2, -1, 0, 1, 2\}$, then zero has only the trivial representation, and for $D_{\min}$ one can take $\tfrac{1}{12} = \min \RR = 0\decdot 1 \overline{2} \, \overline{2} \, \overline{2} \cdots$, where $\overline{d}$ stands for the signed digit $(-d)$.
\end{example}

\begin{remark}
If $\beta >1$ and the alphabet has the form ${\mathcal{A}} = \{0, 1, \ldots, M\}$, then zero has only a trivial representation. But this numeration system has another disadvantage: the operation of subtraction --- which is needed for evaluation of $W_k$ in both multiplication and division algorithms --- is not doable in parallel.
\end{remark}

\begin{example}\label{beta=2}
In the numeration system with $\beta=2$ and ${\mathcal{A}} = \{ -1,0,1 \}$, zero has two non-trivial representations, namely $0 = 0\decdot 1 \overline{1} \, \overline{1} \, \overline{1} \, \overline{1} \cdots = 0\decdot \overline{1} \, 1 \, 1 \, 1 \, 1 \cdots $. Therefore, the preprocessing is a bit more sophisticated than just shifting the fractional point. It is necessary to find a representation of the divisor such that $d_n d_{n+1} \neq 1 \overline{1}$ and $d_n d_{n+1} \neq \overline{1} 1$, where $n$ is the minimal index such that $d_n\neq 0$. This can be achieved by replacing the leading  pair of neighboring digits $1 \overline{1}$ with $01$ or by replacing $\overline{1} 1$ with $0 \overline{1}$, and this procedure is repeated for as long as necessary. Finally, the fractional point is shifted to the first non-zero digit. For example:
$$
0\decdot 1\overline{1}\,\overline{1}\,\overline{1}\,0\overline{1}1001 \mapsto  0\decdot 01\,\overline{1}\,\overline{1}\,0\overline{1}1001 \mapsto 0\decdot 001\,\overline{1}\,0\overline{1}1001 \mapsto 0\decdot 00010\overline{1}1001 \, ,
$$
and lastly, by shifting the fractional point, we get the preprocessed divisor $0\decdot 1 0 \overline{1} 1001$, which can enter as an input of the on-line division algorithm.

The parameter $D_{\min}$ of the Trivedi-Ercegovac algorithm for division can be set to $D_{\min} = \tfrac{1}{4}$ for this numeration system, since any divisor after the described preprocessing satisfies
$$
|D| = |0\decdot d_1 d_2 d_3 \cdots| \geq 0\decdot 1 0 \overline{1} \, \overline{1} \, \overline{1} \cdots = \tfrac12 - \tfrac18 - \tfrac1{16} - \cdots = \tfrac14 = D_{\min} \, .
$$
\end{example}

\begin{example}
In the numeration system with base $\beta =3$ and redundant alphabet ${\mathcal{A}} = \{-1, 0, 1, 2\}$, the number zero has a non-trivial representation $0 = 0\decdot \overline{1} 222 \, \cdots$. In this base, the situation is the same (i.e., a non-trivial $(\beta, {\mathcal{A}})$-representation of zero exists) with any redundant alphabet ${\mathcal{A}}$ containing at least one positive and one negative digit.
\noindent For the numeration system $(3, \{-1, 0, 1, 2\})$, we can set $D_{\min} = \tfrac{1}{9}$, and preprocess by replacing any leading pair of neighboring digits $\overline{1}2$ with $0\overline{1}$, analogously as explained for base $\beta = 2$ in Example~\ref{beta=2}.
\end{example}

We illustrate on two less trivial examples how to find $D_{\min}$ and how to perform preprocessing. In these two examples, the alphabet ${\mathcal{A}}$ consists of (possibly complex) units and zero, and,  moreover, ${\mathcal{A}}$  is closed under multiplication. In order to shorten our list of rules for preprocessing, let us adopt the following conventions:
\begin{enumerate}
    \item instead of the phrase ``If $w_1 w_2 \cdots w_k$ is a prefix of $\bf{d}$, replace this prefix with \\
    $u_1 u_2 \cdots u_k$", we write ``$w_1 w_2 \cdots w_k \longrightarrow u_1 u_2 \cdots u_k$";
    \item the rule ``$w_1 w_2 \cdots w_k \longrightarrow u_1 u_2 \cdots u_k$" is equivalent to the rule ``$w'_1 w'_2 \cdots w'_k \longrightarrow u'_1 u'_2 \cdots u'_k$" if there exists $a \in {\mathcal{A}}, a\neq 0$ such that $w_j = a w_j'$ and $u_j = a u_j'$ for all $j =1, 2, \ldots, k$.
\end{enumerate}
In our list of rules for preprocessing, we mention only one rule from  each class of equivalence. Clearly, each rule on the list preserves the value of the divisors, i.e., $\sum_{j=1}^k w_j \beta^{-j} = \sum_{j=1}^k u_j \beta^{-j}$, and sets $u_1 = 0$. In this convention, the list of preprocessing rules for base $\beta = 2$ and alphabet ${\mathcal{A}} = \{-1, 0, 1\}$ consists of one item only, namely the rule $1 \overline{1} \longrightarrow 01$.

\begin{example}
Let $\beta = \tfrac{1+\sqrt{5}}{2}$ and ${\mathcal{A}} = \{-1, 0, 1\}$. Since $\beta^2 - \beta - 1 = 0$, zero  has the representation $0 = 0\decdot 1 \, \overline{1} \, \overline{1}$.
\noindent We use three preprocessing rules: \quad $ 1) \quad 10\overline{1}\longrightarrow 010$, \quad 2) \quad $1\overline{1}0\longrightarrow001$, \quad 3) \quad  $1\overline{1}\overline{1}\longrightarrow000$.

Let $D = 0\decdot d_1 d_2 d_3 \cdots$. If $d_1 \neq 0$ and none of the rules 1) -- 3) can be applied to the string ${\bf d} = d_1 d_2 d_3 \cdots$, then $|D| \geq  D_{\min} = \tfrac{1}{\beta^5}$. This can be shown by the following analysis, wherein we can assume $d_1 = 1$, without loss of generality:
\begin{itemize}
    \item Let $d_2 = 1$. Then $D \geq \tfrac{1}{\beta} + \tfrac{1}{\beta^2} - \sum_{j= 3}^\infty\beta^{-j} = 1 - \tfrac{1}{\beta} = \tfrac{1}{\beta^2}\geq D_{\min}\,.$
    \item Let $d_2 = 0$. Since the rule 1) cannot be applied, $d_3 \geq 0$, and $D \geq \tfrac{1}{\beta} - \sum_{j= 4}^\infty\beta^{-j} = \tfrac{1}{\beta} - \tfrac{1}{\beta^2} = \tfrac{1}{\beta^3}\geq D_{\min}\,.$
    \item Let $d_2 = \overline{1}$. As the rules 2), 3) cannot be applied, $d_3 = 1$. Thus $D \geq \tfrac{1}{\beta} - \tfrac{1}{\beta^2}  +\tfrac{1}{\beta^3} - \sum_{j= 4}^\infty\beta^{-j} = \tfrac{1}{\beta^5} = D_{\min} \, .$
\end{itemize}
\end{example}

\begin{example}\label{EisenPrepro}
Let $\beta = -1 + \omega \,$, where $\omega = \exp{\frac{2i\pi}{3}}$ is the third root of unity, i.e., $\omega^3 = 1$. We consider the alphabet ${\mathcal{A}}$ of size $\#{\mathcal{A}} = 7$, namely ${\mathcal{A}} = \{0, \pm 1, \pm \omega, \pm \omega^2 \}$. Section~\ref{Eisen} is devoted to this numeration system in detail; here we just mention that the elements of the ring $\Z[\omega]$ are called Eisenstein integers.

Firstly, we show that
\begin{equation}\label{Dmax}
    D_{\max} = \max \{| 0 \decdot d_1 d_2 d_3 \cdots | \ :  d_j \in {\mathcal{A}}\} = \tfrac12 \sqrt{7}\,.
\end{equation}
Since $|x\beta + y| \leq |\beta -1| = \sqrt{7}$ for any $x, y \in{\mathcal{A}}$, we have $$|0\decdot d_1 d_2 d_3 \cdots| \leq \sqrt{7}\sum\limits_{j= 1}^{\infty}|\beta|^{-2j}= \tfrac12 \sqrt{7}.$$
As $\beta^2 = a|\beta^2|$, where $a =-\omega \in {\mathcal{A}}$, the upper bound $\tfrac12 \sqrt{7}$ is attained.

Let us list 9 equivalence classes of the rules that we apply in the (divisor) preprocessing:
\begin{itemize}
    \item Using $\beta + (1-\omega) = 0$, we get the rules \\
        $\quad A) \ \ 1 \, 1 \longrightarrow 0\, \omega$,
        $\qquad \ B) \ \ 1 \, \overline{\omega} \longrightarrow 0\, \overline{1}$.
    \item From $\beta^{2} + \beta + (\omega-\omega^2)= 0$, we obtain \\
        \ \quad C) \ \ $1 \, 0 \, \omega \longrightarrow 0\, \omega \,1$, \qquad \ \, D) \ \ $1 \, 0 \, \overline{\omega^2} \longrightarrow 0\, \overline{1}\, \overline{\omega}$, \\ \qquad E) \ \ $1 \, \overline{\omega^2} \, \omega \longrightarrow 0\, \omega\, \omega^2$, \quad \quad F) \ \ $1 \, \overline{\omega^2} \, \overline{\omega^2} \longrightarrow  0\, \omega\, \overline{\omega}$.
    \item Using $\beta^{2} - \omega \beta + (\omega - 1) = 0$, we get the rules\\
        \quad G) \ \ $1 \, 0 \, \overline{1}\longrightarrow 0\, \omega\, \overline{\omega}$, \ \quad \ H) \ \ $1 \, \omega^2 \, \overline{1} \longrightarrow 0\, \overline{1} \, \overline{\omega}$, \\\qquad  I) \ \ $1 \, \omega^2 \, \omega \longrightarrow 0\, \overline{1}\, 1$.
\end{itemize}
If $D = 0\decdot d_1 d_2 d_3 \cdots$ with $d_1 \neq 0$, and none of the rules A) -- I) can be applied to the string ${\bf d} = d_1 d_2 d_3 \cdots$, then
\begin{equation}\label{Eisenstein_D_min}
    |D| \geq D_{\min} = \frac{\sqrt{3} (6-\sqrt{7})}{18} \, .
\end{equation}
Without loss of generality, we can assume $d_1 = 1$. By exploring all possible triplets $1 d_2 d_3$ to which no rules can be applied, we see that $|0\decdot  d_1d_2d_3| \geq \tfrac{\sqrt{3}}{3}$. Therefore, $|D| \geq \tfrac{\sqrt{3}}{3} - \tfrac{1}{|\beta|^3}D_{\max}$, which, together with~\eqref{Dmax}, proves~\eqref{Eisenstein_D_min}.
\end{example}

\noindent

Let us conclude this section by three remarks concerning the optimality of the parameters occurring in the on-line algorithms.

\begin{remark}
Note that the preprocessing methods and results given in examples above may not be optimal, in the sense that the values $D_{\min}$ may not be the maximal possible. Some of them could be further increased, by performing more laborious preprocessing,  especially by deploying larger sets of rewriting rules.   In general, the bigger  the value $D_{\min}$ is the smaller the delay $\delta$ can be used in the on-line division algorithm.
\end{remark}

\noindent
\begin{remark}  To show the correctness of the Trivedi-Ercegovac algorithm we did not need  the inequality $|D_{k}| \geq D_{min}$ to be valid    for all $k \geq 1$. In the division algorithm,   the  select function  is applied only to  the value  $D_k$ for $ k=\delta +1, \delta +2, \ldots $ , see  \eqref{eq:W_k_div}.  For such indices $k$ we also  required  $|\Trunc_\alpha(D_k)| \geq D_{\min}$, see
\eqref{trunc_alpha_D}.  The function  $\Trunc_\alpha$  uses only $L$ fractional digits of a string representing its argument.  Such $L$    depends on $\alpha$  and $A =\max\{|a| : a \in \mathcal{A}\}$, and  must be chosen to   satisfy  $ \sum_{k\geq L+1}A |\beta|^{-k}  <\alpha$ .   Clearly,   $|\Trunc_\alpha(D_k)|= |D_L|$,   for all  $k\geq L$.

Let us summarize: The correctness of the Trivedi-Ercegovac division algorithm requires $|D_L| \geq D_{\min}$  and   $|D_{k}| \geq D_{\min}$  for all $  k\geq \delta +1 $.
\end{remark}


\begin{remark}\label{OL_2in1}
Definition~\ref{DefOL} of the OL Property covers  in fact two purposes:
\begin{enumerate}
    \item boundedness of the sequences $(W_k)$, so that on-line multiplication and division algorithms converge;
    \item  sufficiency of  using only truncated representations of $W_k$ and $D_{k + \delta}$, which is necessary for a
cheap evaluation  of the $\Select$ functions (the question of complexity of the algorithms is discussed in the next section).
\end{enumerate}
\noindent
Looking into  the proofs  of correctness of the algorithms one can see that  these two purposes are reflected in the OL Property definition by:
\begin{enumerate}
    \item covering  of the $\varepsilon$-fattening of the set $(\beta I)$ by the union of sets $\bigcup_{a \in {\mathcal{A}}} (I + a)$;
\item each point $x$ of $(\beta I)^{\varepsilon}$  sits  inside a  set $I+a$  and  deep behind its border, or more precisely  the distance between  $x$  and the border   of  $I+a$ is at least $\varepsilon$.
 \end{enumerate}
To avoid  very   technical formulation of  \eqref{OL} we decided to use the same parameter $\varepsilon$ to take into account both phenomena.

For a finer calculation of parameters $\delta$ and $L$,  it may be useful to parameterize these two aspects separately. It means  to use one parameter $\mu > 0$ for a  fattening of the set $(\beta I)$  and another parameter $\nu> 0$ for watching the distance to the border  of $I+a$.
This approach was used for the  Eisenstein numeration system, see Section~\ref{Eisen}.
\end{remark}

\section{Time complexity of  the  Trivedi-Ercegovac  algorithms}

The time complexity of an algorithm is usually defined as the number of elementary operations  needed to get a result for any input of length  $n$. In our multiplication and division algorithms, strings representing
input numbers can be infinite. Therefore, by time complexity $T(n)$ we understand the number of elementary   operations needed to get $n$ digits of the result on the output of the algorithms.  The time complexity of both  algorithms depends on the number of steps needed to compute the auxiliary value $W_k$ and  the $k$-th output digit  by the relevant $\Select$ function. If both tasks can be performed in constant time, then the time complexity of computing the first $n$ most significant digits of the result is $\mathcal{O}(n)$.

\subsection{Evaluation of $W_k$}
According to Formulas \eqref{eq:N6} and \eqref{W_k_definition}, the values of $W_k$   can be calculated in constant time if addition and subtraction and also multiplication by a digit from ${\mathcal{A}}$ can be performed in parallel in $(\beta, {\mathcal{A}})$. It is possible only in a redundant numeration system $(\beta, {\mathcal{A}})$.

Already  the OL Property  forces the system to be redundant. For real bases, redundancy implies $\#{\mathcal{A}} > |\beta|$, and Lemma \ref{intervalI} states that $\#{\mathcal{A}} > |\beta|$ is also a sufficient condition for the OL Property in case of an alphabet of contiguous integers. Nevertheless, $\#{\mathcal{A}} > |\beta|$ does not guarantee that addition and subtraction in $(\beta, {\mathcal{A}})$ are doable in constant time in parallel. Usually, the alphabet has to be extended further on. For example, both systems $(\tfrac{1+\sqrt{5}}{2}, \{0,1\})$ and $(\tfrac{1+\sqrt{5}}{2}, \{\overline{1},0,1\})$ have the OL Property, but parallel addition and subtraction is possible only in the second one. The question of sufficient redundancy for parallel addition is treated in general in~\cite{FrPeSv2}.

\subsection{Evaluation of the Select function}
To evaluate the $\Select$  functions in constant time, their  output values  $p_k$ and $q_k$  must depend only on a bounded number of digits in the strings representing the variables $W_k$  and $D_k$.
 If it is the case, then evaluation of the  $\Select$  function is performed  by using finite look-up tables, as  on a finite alphabet there exist only finitely many strings  of bounded length.
Let us concentrate on this case and consider strings representing  $W_k$  and $D_k$.

From the right side --- i.e., behind the fractional point --- the number of fractional digits of $W_k$ and $D_k$ is limited to $L \in \N$ by truncation (see Definition \ref{truncation}) of the less significant digits in the $(\beta, A)$-representation of $W_k$ and $D_k$:
\begin{itemize}
    \item for multiplication: $\Trunc_{\varepsilon/2}(W_k) = \sum_{j=-n}^{L} w_j \beta^{-j}$, so that $|\sum_{j=L+1}^\infty w_j \beta^{-j}| < \varepsilon/2$ with $\varepsilon$ from the OL Property;
    \item for division: $\Trunc_{\alpha}(W_k) = \sum_{j=n}^{L} w_j \beta^{-j}$ and $\Trunc_{\alpha}(D_k) = \sum_{j=1}^{L} d_j \beta^{-j}$, so that $|\sum_{j=L+1}^\infty w_j \beta^{-j}| < \alpha$ and $|\sum_{j=L+1}^\infty d_j \beta^{-j}| < \alpha$ with $\alpha$ defined in \eqref{ALFA}.
\end{itemize}
The parameter $L$ is found simply by solving the inequalities (separately for multiplication or for division): let $A = \max\{ |a| \,:\, a \in {\mathcal{A}}\}$
\begin{equation}\label{L_calculation}
    \tfrac{A}{|\beta|^{L+1}} + \tfrac{A}{|\beta|^{L+2}} + \cdots = \tfrac{A}{|\beta|^L(|\beta|-1)} < \tfrac{\varepsilon}{2} \quad \text{ or } \quad \tfrac{A}{|\beta|^L(|\beta|-1)} < \alpha \, .
\end{equation}

To  limit also the number of digits before the fractional point  of $W_k$  we use  the fact  that  $W_k$  belongs to a bounded area, say  $\mathcal{J}$.   For $L \in \N$ and a bounded set $\mathcal{J}$, let us consider the following set of strings over ${\mathcal{A}}$:
\begin{equation}\label{S_definition}
    \S_{L,\mathcal{J}} = \Bigl\{ x_{-n} x_{-n+1} \cdots x_{-1} x_0 x_{1} x_{2} \cdots x_{L} \, : \, x_{-n} \neq 0 \, , x_j \in {\mathcal{A}} \ \text{ and } \ \sum_{j=-n}^L x_j \beta^{-j} \in \mathcal{J} \Bigr\} \, .
\end{equation}

\begin{lemma}\label{finitBeforeFP}
Let $(\beta, {\mathcal{A}})$ be a   numeration system,  $\mathcal{J}$  a bounded set and $L \in \N$. If zero has only the trivial $(\beta, {\mathcal{A}})$-representation, then the set $\S_{L,\mathcal{J}}$ from \eqref{S_definition} is finite.
\end{lemma}

\begin{proof}
Assume that $\S_{L, \mathcal{J}}$ is infinite. Since ${\mathcal{A}}$ is finite, there exist a strictly increasing sequence of  positive integers $k_n$ and a sequence $x^{(n)}$ of numbers $x^{(n)} = \sum _{j=-k_n}^{L} x^{(n)}_j \beta^{-j}$ such that $x^{(n)}_{-k_n} \neq 0$ and $x^{(n)} \in \mathcal{J}$. It implies that $|x^{(n)} \beta^{-k_n-1}| \in \RR$, as defined in \eqref{defR}. Since $\mathcal{J}$ is bounded, say by a constant $b$ (in modulus), we have $|x^{(n)} \beta^{-k_n-1}| \leq b | \beta^{-k_n-1}| \to 0$, and thus the infimum of ${\mathcal R}$ is zero --- a contradiction with Lemma \ref{R}.
\end{proof}

   If zero has only the trivial $(\beta, {\mathcal{A}})$-representation,   we create a look-up table
for multiplication by considering   $\mathcal{J} = (\beta I)^{\varepsilon/2}$.  For division, we  consider all possible truncated divisors  $D = 0\decdot d_1 \cdots d_L$, with $d_1\neq 0$  and create  a look-up table  for the  bounded set  $\mathcal{J}_D = D (\beta I)^{\varepsilon/2}$.

Even if  the set $\S_{L,\mathcal{J}}$ is not finite,  the situation may be not hopeless.  If the numeration system allows  preprocessing of divisor, one can use its  rewriting  rules   (without shifting the fractional point) and  modify
 the representation of $W_k$ after each iterative step of the algorithm.  So we prevent to have  a  representation of $W_k$ with excessive index of the leading coefficient and thus only finitely many strings represent  all possible values occurring in the truncated sequence  $(W_k)$.

\begin{lemma}
Let us assume that for a numeration system $(\beta, {\mathcal{A}})$ there exists a finite list of rules and $D_{\min} > 0$ such that any $D = 0\decdot  d_1 d_2 d_3 \cdots$ with $d_1 \neq 0$ on which no rule of the list can be applied has modulus $|D|\geq D_{\min}$. Then the set $\S'_{L,\mathcal{J}} = \{s \in \S_{L,\mathcal{J}} : \text {no rules can be applied to the string} \ s \}$ is finite, for a given integer $L \in \N$ and a given bounded set $\mathcal{J}$.
\end{lemma}

\begin{proof}
It is analogous to the proof of Lemma \ref{finitBeforeFP}, with only $\inf {\mathcal R}$ replaced by $D_{\min}$.
\end{proof}

We can summarize the findings from this section into the following statement:
\begin{theorem}
If a numeration system $(\beta, {\mathcal{A}})$ with the OL Property allows parallel addition, parallel subtraction, and preprocessing of divisors into the form \eqref{Denominator}, then the time complexity of the Trivedi-Ercegovac algorithms for on-line multiplication and division is $\mathcal{O}(n)$.
\end{theorem}


\section{Examples}\label{examples}


\subsection{Base $\beta = \frac{3+\sqrt{5}}{2}$ and alphabet ${\mathcal{A}} =\{-1,0,1\}$ }\label{goldenSquar}

For illustrating how the on-line algorithms for multiplication and division work, we consider
a well-studied numeration system, with base $\beta = \frac{3+\sqrt{5}}{2}$ and alphabet ${\mathcal{A}} =\{-1,0,1\}$. Let us list the most important properties of this system:
\begin{itemize}
    \item The base $\beta = \frac{3+\sqrt{5}}{2}$ is a quadratic Pisot unit with minimal polynomial $f(t) = t^2 -3t+1$. In fact, $\beta$ is the square of the golden mean $\frac{1+\sqrt{5}}{2}$.
    \item The numeration system with base $\beta = \frac{3+\sqrt{5}}{2}$ and alphabet ${\mathcal{A}} =\{-1,0,1\}$ allows parallel addition \cite{FrPeSv2}.
    \item By Lemma \ref{ReprezZero}, zero has only a trivial $(\beta, {\mathcal{A}})$-representation, and
            \begin{equation}\label{sign}
                D_{\min} = 0\decdot\, 1 \, \overline{1} \, \overline{1}\,\overline{1}\,\cdots = \frac{1}{\beta} - \sum_{j=2}^\infty \frac{1}{\beta^j} = \frac{1}{\beta^2} \,.
            \end{equation}
            It means that the sign of the first non-zero digit in a representation decides about the sign of the represented number. Moreover, preprocessing of divisor consists just in shifting the fractional point.
    \item If $D =0\decdot d_1d_2d_3\cdots$ is a $(\beta, {\mathcal{A}})$-representation of the number $D$, then
            \begin{equation}\label{estimate}
                -\frac{1}{\beta-1} \le D \le \frac{1}{\beta-1}=D_{\max}.
            \end{equation}
    \item By Lemma \ref {intervalI}, the numeration system has the OL Property with
            \begin{equation}\label{eps_tau^2}
                \varepsilon = \frac{1}{2\beta(\beta+1)} > 0 \, \quad \text{and} \quad  I=[-\rho, \rho],  \quad \text{where}\ \rho = \tfrac12 + \varepsilon = \frac{2}{\beta+1}\,.
            \end{equation}
    \item The $\Digit$ function from Lemma \ref{digit} is $\Digit: [-\beta\rho -\varepsilon, \beta\rho + \varepsilon] \to \{-1, 0, 1\}$ defined by
            \begin{equation}\label{alter}
                \Digit(V) = \left\{\begin{array}{lll}
                    \ \ 1 & & \text{if \ \ } V > \tfrac12 \, ,\\
                    -1 & & \text{if \ \ } V < -\tfrac12 \, ,\\
                    \ \ 0 & & \text{otherwise.}
                \end{array} \right.
            \end{equation}
\end{itemize}


\subsubsection{On-line multiplication in base $\beta = \frac{3+\sqrt{5}}{2}$ and alphabet ${\mathcal{A}} = \{-1, 0, 1\}$}\label{specificBeta}

For on-line multiplication, the delay $\delta$ according to \eqref{deltatMult} has to satisfy $\tfrac{2}{\beta^\delta(\beta -1)}< \tfrac{1}{4\beta(\beta +1)}$, and the smallest such delay is $\delta = 4$.

It remains to find an easy way how to evaluate the function $\SelectM(W) = \Digit(\Trunc_{\varepsilon/2}(W))$. By Definition~\ref{SelectMult}, its domain is $(\beta I)^{\varepsilon/2} = [-\beta \rho - \tfrac12 \varepsilon, \beta \rho + \tfrac12 \varepsilon]$.\\

\noindent{\bf Claim 1}: If $Z = z_{-m} z_{-m+1} \cdots z_{-1} z_0 \decdot z_{1} z_{2} \cdots \in (\beta I)^{\varepsilon/2}$ with $z_{-m} \neq 0$, then $-m\leq 1$. Without loss of generality, consider $z_{-m} =1$. For contradiction, suppose that $-m \geq 2$. Then
$$
Z \geq \beta^{-m} - \sum_{j=-m+1}^{+\infty} \beta^{-j} = \beta^{-m} \Bigl(1 - \tfrac{1}{\beta-1}\Bigr) \geq \beta^2 \Bigl(1 - \tfrac{1}{\beta-1}\Bigr) = \beta > \beta \rho + \tfrac12 \varepsilon,
$$
i.e., $Z \notin (\beta I)^{\varepsilon/2}$ --- a contradiction.\\

\noindent{\bf Claim 2}: Let $Z = z_{-1} z_0 \decdot z_{1} z_{2} \cdots$ and $V = z_{-1} z_0 \decdot z_{1} z_{2} z_{3} z_{4}$. Then $|Z - V| < \tfrac12 \varepsilon$, with $\varepsilon$ defined in \eqref{eps_tau^2}.\\
Indeed, $|Z-V| \leq \sum_{j\geq 5}\beta^{-j} = \tfrac{1}{\beta^4}\tfrac{1}{\beta-1} < \tfrac12 \varepsilon \, .$\\

\noindent{\bf Claim 3}: Let $V = z_{-1} z_0 \decdot z_{1} z_{2} z_{3} z_{4}$. Then $V > \tfrac12$ if and only if
$$
z_{-1} z_0 z_{1} z_{2} z_{3} \succ 0 \, 1 \, \overline{1} \, \overline{1} \, 0 \quad \text{ or} \quad \bigl(z_{-1} z_0 z_{1} z_{2} = 0 \, 0 \, 1 \, 1 \quad \text{and} \quad z_{3} \neq \overline{1}\bigr)
$$
where $\succ$ denotes the lexicographic order on words over the alphabet ${\mathcal{A}} = \{-1, 0, 1\}$.

This can be proved by inspection of all possibilities, and using the symmetry of the alphabet and \eqref{sign} as follows:
\begin{itemize}
    \item If $z_{-1} = 1$, then $V > \beta - \sum_{j= 0}^\infty \beta^{-j} = 1 > \tfrac12$.
    \item If $z_{-1} = 0, z_0 = 1, z_{1} \geq 0$, then $V \geq 1- \sum_{j= 2}^\infty \beta^{-j} = \tfrac{2}{\beta} > \tfrac12$.
    \item If $z_{-1} = 0, z_0 = 1, z_{1} = \overline{1}, z_{2} \geq 0$, then $V \geq 1 - \tfrac{1}{\beta} - \sum_{j= 3}^\infty \beta^{-j} = \tfrac{1}{\beta} + \tfrac{1}{\beta^2} > \tfrac12$.
    \item If $z_{-1} = 0, z_0 = 1, z_{1} = \overline{1}, z_{2} = \overline{1}, z_{3} = 1$, then $V \geq 1 - \tfrac{1}{\beta} - \tfrac{1}{\beta^2} + \tfrac{1}{\beta^3} - \tfrac{1}{\beta^4} > \tfrac12$.
    \item If $z_{-1} = 0, z_0 = 1, z_{1} = \overline{1}, z_{2} = \overline{1}, z_{3} \leq 0$, then $V \leq 1- \tfrac{1}{\beta} - \tfrac{1}{\beta^2} + \tfrac{1}{\beta^4} < \tfrac12$.
    \item If $z_{-1} = 0$ and $z_0 = 0$, we have to perform a more detailed  calculation:
        \begin{itemize}
             \item as $0\decdot 110\overline{1} > \tfrac12$, the value $0\decdot 11 z_{3} z_{4} > \tfrac12$ for any $z_{3} \geq 0$ and $z_{4} \in {\mathcal{A}}$;
             \item as $0\decdot 11\overline{1}1 < \tfrac12$, the value $0\decdot 11 \overline{1} z_{4} < \tfrac12$ for any $z_{4} \in {\mathcal{A}}$;
             \item as $0\decdot 1011 < \tfrac12$, the value $0\decdot 10 z_{3} z_{4} < \tfrac12$ for any $z_{3},z_{4} \in {\mathcal{A}}$;
             \item as $0\decdot 0111 < \tfrac12$, the value $0\decdot 0 z_{2} z_{3} z_{4} < \tfrac12$ for any $z_{2}, z_{3}, z_{4} \in {\mathcal{A}}$.
        \end{itemize}
    \item If $z_{-1} = 0$ and $z_0 = \overline{1}$, then $V \leq -1+ \sum_{j= 1}^\infty \beta^{-j}= -1 + \tfrac{1}{\beta - 1} < \tfrac12 \, .$
\end{itemize}

Lemma \ref{finitBeforeFP} guarantees that the evaluation of the function $\SelectM$ can be done via a finite table of values. Previous Claims 1--3 imply that such table has $3^5$ elements (i.e., 3 possible digits from ${\mathcal{A}}$ on 5 positions $z_{-1}, \ldots, z_3$). But Claim 3 and the lexicographic order enable us to provide also a more effective way of evaluation of the function $\SelectM$.

Let $Z = z_{-m} z_{-m+1} \cdots z_{-1} z_0 \decdot z_{1} z_{2} \cdots$  be a $(\beta, {\mathcal{A}})$-representation of a number $Z$ in base $\beta = \frac{3+\sqrt{5}}{2}$ and alphabet ${\mathcal{A}} = \{-1, 0, 1\}$. We define
\begin{equation}
    \SelectM(Z) = \left\{\begin{array}{ll}
        \ \ 1 &  \text{if}  \; z_{-1} z_0 z_{1} z_{2} z_{3} \succ 0 \,1\,\overline{1}\, \overline{1} \,0\ \; \text{or}  \; \bigl(z_{-1} z_0 z_{1} z_{2} = 0\,0\,1\,1 \; \text{and} \; z_{3} \neq \overline{1}\bigr) ,\\
        -1 & \text{if}  \; z_{-1} z_0 z_{1} z_{2} z_{3} \prec 0 \,\overline{1}\,{1}\, {1} \,0\ \;  \text{or} \;  \bigl(z_{-1} z_0 z_{1} z_{2} = 0\,0\,\overline{1}\,\overline{1} \; \text{and} \; z_{3} \neq 1\bigr) ,\\
        \ \ 0 & \text{otherwise}.
    \end{array} \right.
\end{equation}

In base $\beta = \frac{3+\sqrt{5}}{2}$ with alphabet ${\mathcal{A}} = \{-1,0,1\}$, on-line multiplication is possible by the Trivedi-Ercegovac algorithm with delay $\delta = 4$, and with linear time complexity. The number of digits we need to evaluate for $W$ within the algorithm is $L = 3$ behind the fractional point, and another two digits before the fractional point.


\subsubsection{On-line division in base $\beta = \frac{3 + \sqrt{5}}{2}$ and alphabet ${\mathcal{A}} = \{-1, 0, 1\}$}\label{specificBeta}

To determine the algorithm for on-line division in base $\beta = \frac{3 + \sqrt{5}}{2}$, we have to specify two parameters: $\delta$ and $D_{\min}$. We put   $D_{\min} = \frac{1}{\beta^2}$ (cf. \eqref{sign}). To find the delay $\delta$, we may again follow the general formula \eqref{DeltaDivision}, and obtain $\delta = 7$. By a more elaborated calculation, specific for this numeration system, the delay can be further optimized, namely to $\delta = 6$, in combination with the number $L=9$ of fractional digits to evaluate in the representations of $W$ and $D$.

In the sequel, we show that the delay can be set to $\delta = 6$. Since we work with a symmetric alphabet, we assume in the whole section that the denominator is positive, i.e., its first digit $d_1 =1$.

We start with two auxiliary claims, using $(\beta I)^{\varepsilon/2} = [-\beta \rho - \tfrac12 \varepsilon, \beta \rho + \tfrac12 \varepsilon]$.\\

\noindent{\bf Claim 1}: Let $Z = z_{-m} z_{-m+1} \cdots z_{-1} z_0 \decdot z_{1} z_{2} \cdots \in D (\beta I)^{\varepsilon/2}$ with $z_{-m} \neq 0$ and $D_{\min} = \tfrac{1}{\beta^2} < D < \tfrac{1}{\beta-1} = D_{\max} $. Then $m \leq 0$.

\noindent Without loss of generality, consider $z_{-m} =1$. For contradiction, suppose that $m \geq 1$. Then
$$
Z \geq \beta^m - \sum_{j=-\infty}^m \beta^j = \beta^m \Bigl(1-\tfrac{1}{\beta-1}\Bigr) \geq \beta\Bigl(1-\tfrac{1}{\beta-1}\Bigr) = 1 > \tfrac{1}{\beta-1} \Bigl(\beta \rho + \tfrac12 \varepsilon\Bigr) = D_{\max} \Bigl(\beta \rho + \tfrac12 \varepsilon\Bigr) \, ,
$$
i.e., $Z \notin D(\beta I)^{\varepsilon/2} $ --- a contradiction.\\

\noindent{\bf Claim 2}: Let $U = u_0 \decdot u_{1} u_{2} \cdots$ and $D = 0 \decdot d_1 d_2 \cdots$. Denote $\Delta = 0 \decdot d_1 d_2 \cdots d_9$ and $V = u_0 \decdot u_{1} u_{2} \cdots u_{9}$. Then
$$
\Bigl|\tfrac{U}{D}-\tfrac{V}{\Delta}\Bigr| < \tfrac{\varepsilon}{2}\,.
$$
Indeed, find $\alpha_1$ and $\alpha_2$ such that $U = V + \alpha_1$ and $D = \Delta + \alpha_2$. The moduli of $\alpha_1$ and $\alpha_2$ are bounded by $\sum_{j>9}\beta^{-j} = \tfrac{1}{\beta^9(\beta-1)}$. Using \eqref{sign} and \eqref{estimate}, we get
$$
\Bigl|\tfrac{U}{D}-\tfrac{V}{\Delta}\Bigr| = \tfrac{1}{D\Delta}\bigl| \alpha_1\Delta - \alpha_2V\bigr| \leq  \tfrac{1}{D_{\min}^2} \tfrac{1}{\beta^9(\beta-1)} (1+\beta) D_{\max}= \tfrac{1+\beta}{\beta^5(\beta-1)^2}< \tfrac12\varepsilon \, .
$$

By combining the previous Claims 1--2 and the form of the $\Digit$ function given by \eqref{alter}, we define the $\SelectD$ function for on-line division.

Let $U = u_{-n} u_{-n+1} \cdots u_{-1} u_0 \decdot u_{1} u_{2} \cdots$ and $D = 0\decdot d_1 d_2 \cdots$, where $d_1 = 1$; and denote
\begin{eqnarray}\label{selectDivStach}
\Delta = 0\decdot d_1 d_2 \cdots d_9 \quad \text{and} \quad V = u_0 \decdot u_{1} u_{2} \cdots u_{9} \\
\SelectD(U,D) = \left\{\begin{array}{lll}
    \ \ 1 & & \text{if \ \ } 2V - \Delta > 0 \, , \\
       -1 & & \text{if \ \ } 2V + \Delta < 0 \, ,\\
    \ \ 0 & & \text{otherwise.}
\end{array} \right.
\end{eqnarray}
Since $V$ and $\Delta$ use only a limited number of digits, the values $2V - \Delta$ and $2V + \Delta$ are computable in constant time. In our numeration system, $0$ has only trivial representation, therefore the most significant digit of $2V - \Delta$ and $2V + \Delta$ decides about positivity or negativity. Consequently, $\SelectD$ can be evaluated in constant time.\\

\noindent{\bf Claim 3}: If $\tfrac{U}{D}\in (\beta I)^{\varepsilon/2}$, then $\tfrac{U}{D} - \SelectD(U,D) \in [-\rho +\tfrac12\varepsilon, \rho -\tfrac12\varepsilon]$.

As $\tfrac{U}{D} \in [-\beta \rho - \tfrac12 \varepsilon, \beta \rho + \tfrac12 \varepsilon]$, and by virtue of Claim 2, we have in particular
\begin{equation*}\label{test}
    \tfrac{V}{\Delta} - \tfrac12\varepsilon < \tfrac{U}{D} \leq \beta \rho + \tfrac12 \varepsilon\,.
\end{equation*}
In the sequel, we exploit the fact that we determined  $\rho$ and $\varepsilon$ by \eqref{eps_tau^2}. Our discussion is split into three cases, according to the value $q=\SelectD(U,D) \in \{-1, 0, 1\}$.
\begin{description}
    \item[ $q=1$\,:] \quad By~\eqref{selectDivStach}, we have $\tfrac{V}{\Delta} > \tfrac12$. Thus
        $$
        -\rho +\tfrac12 \varepsilon = -\tfrac12\varepsilon-\tfrac12 < \tfrac{V}{\Delta} -\tfrac12\varepsilon-1 < \tfrac{U}{D} -1 \leq \beta \rho + \tfrac12\varepsilon-1 < \rho - \tfrac12 \varepsilon \, .
        $$
    \item[ $q=0$\,:] \quad Then $-\tfrac12 \leq \tfrac{V}{\Delta} \leq \tfrac12$. Consequently,
        $$
        \tfrac{U}{D} - 0 < \tfrac{V}{\Delta} + \tfrac12 \varepsilon \leq \tfrac12 + \tfrac12 \varepsilon = \rho - \tfrac12 \varepsilon \, .
        $$
        The lower bound for $\tfrac{U}{D} - 0$ and the whole case $q = \overline{1}$ follow by symmetry.
\end{description}

\noindent{\bf Claim 4}: Let $J = (\beta I)^{\varepsilon/2} = [-\beta \rho - \tfrac12 \varepsilon, \beta \rho + \tfrac12 \varepsilon]$. Then $W_0 = 0 \in J$, and for any $k \in \N$ the implication $W_{k} \in D_{k+\delta} J \Rightarrow W_{k+1} \in D_{k+\delta+1} J$ holds. Consequently, $(W_k)$ is bounded if the delay is at least $\delta = 6$.

According to \eqref{W_k_definition}, we have
$$
W_{k+1} = \beta \bigl(W_{k} - q_{k} D_{k+\delta}\bigr) + (n_{k+1+\delta} - Q_{k} d_{k+1+\delta}) \beta^{-\delta} \, ,
$$
where $q_k = \SelectD(W_{k}, D_{k+\delta})$. We give upper bounds on the two previous summands separately. Firstly,
\begin{equation}\label{second}
    \bigl| n_{k+1+\delta} - Q_{k} d_{k+1+\delta} \bigr| \beta^{-\delta} \leq (1+D_{\max}) \beta^{-\delta} \, .
\end{equation}

Secondly, we apply Claim 3, and, due to $D_{k+1+\delta} = D_{k+\delta} + \tfrac{d_{k+1+\delta}}{\beta^{k+1+\delta}} \geq D_{k+\delta} - \tfrac{1}{\beta^{1+\delta}}$, we have
\begin{equation}\label{third}
    \beta \bigl| W_{k}- q_{k} D_{k+\delta} \bigr| \leq \beta \bigl| D_{k+\delta} \bigr| \bigl( \rho - \tfrac12 \varepsilon\bigr) \leq \bigl| D_{k+1+\delta} \bigr| \beta \bigl(\rho -\tfrac{\varepsilon}{2} \bigr) + (\rho -\tfrac12\varepsilon)\beta^{-\delta}.
\end{equation}

Suppose that the inequality
\begin{equation}\label{fourth}
    (1 + D_{\max}) \beta^{-\delta} + (\rho - \tfrac12\varepsilon)\beta^{-\delta} < D_{\min} (\tfrac12\beta \varepsilon + \tfrac12\varepsilon)
\end{equation}
is satisfied; then, by adding inequalities \eqref{second} and \eqref{third}, we obtain
\begin{eqnarray*}
    \bigl| W_{k+1} \bigr| & < & \bigl| D_{k+1+\delta} \bigr| \bigl(\beta\rho - \tfrac12 \beta \varepsilon \bigr) + D_{\min} (\tfrac12 \beta \varepsilon + \tfrac12 \varepsilon) \leq \\
    & \leq & \bigl| D_{k+1+\delta} \bigr| \bigl( \beta \rho - \tfrac12 \beta \varepsilon \bigr) + \bigl| D_{k+1+\delta} \bigr| (\tfrac12\beta \varepsilon + \tfrac12 \varepsilon) = \\
    & = & \bigl| D_{k+1+\delta} \bigr| \bigl(\beta\rho +\tfrac12\varepsilon\bigr) \, .
\end{eqnarray*}

Due to the symmetry of the interval $J$ with respect to $0$, we have $\bigl| D_{k+\delta+1} \bigr| J = D_{k+\delta+1} J$, and thus $W_{k+1} \in \bigl| D_{k+\delta+1} \bigr| J = D_{k+1+\delta} J$. A simple calculation shows that \eqref{fourth} is satisfied if the delay is at least $\delta = 6$.\\

We can summarize: In base $\beta = \frac{3 + \sqrt{5}}{2}$ with alphabet ${\mathcal{A}} = \{-1, 0, 1\}$, on-line division is possible by the Trivedi-Ercegovac algorithm with delay $\delta = 6$, and with linear time complexity. The number of fractional digits to evaluate for $W$ and $D$ within the algorithm is $L = 9$, and another digit before the fractional point for~$W$.


\subsection{Knuth numeration system}

D. E. Knuth showed in 1955 \cite{Knuth} that in the numeration system with base $\beta =2i$ and alphabet $\{0,1,2,3\}$, any complex number $Z$ has a representation of the form $Z = \sum_{j\geq n} z_j \beta^{-j}$, where $n \in \Z$ and $z_j \in \{0, 1, 2, 3\}$. In this numeration system, almost all complex numbers have a unique representation. We consider a redundant system with the same base and a symmetric alphabet ${\mathcal{A}} = \{-2, -1, 0, 1, 2\}$. Let us list the relevant properties of this system:
\begin{itemize}
    \item In $(\beta, {\mathcal{A}})$, parallel addition is possible, see \cite{Frougny2}.
    \item The system $(\beta, {\mathcal{A}})$ has the OL Property, as the oblong $I$ with vertices $\pm \tfrac59 \pm i \tfrac{11}9$ and $\varepsilon = \tfrac{1}{18}$ satisfies \eqref{OL}.
    \item The function $\Digit$ is defined by
        \begin{equation*}
        \Digit(V)= \left\{\begin{array}{lll}
            \ \ 2 & & \text{if \ \ } \Re(V) > \tfrac32 \, , \\
            \ \ 1 & & \text{if \ \ } \Re(V) \in (\tfrac12,\tfrac32] \, , \\
            \ \ 0 & & \text{if \ \ } \Re(V) \in [-\tfrac12,\tfrac12] \, , \\
            -1 & & \text{if \ \  } \Re(V) \in [-\tfrac32,-\tfrac12) \, , \\
            -2 & & \text{if \ \ } \Re(V) < -\tfrac32 \, .
        \end{array} \right.
        \end{equation*}
    \item A number $Z = \sum_{j=1}^{\infty} z_{j} \beta^{-j}$ with $z_j \in {\mathcal{A}}$ can be decomposed into real and imaginary part as follows:
        $$
        Z = \sum_{j=1}^{\infty} z_{j} \beta^{-j} =
        \sum_{j=1}^{\infty} z_{2j}(-4)^{-j} + 2i\sum_{j=1}^{\infty} z_{2j-1}(-4)^{-j} \, .
        $$
        It means that each of the real and imaginary parts can be represented in the real numeration system with base $-4$ and alphabet $\{-2, -1, 0, 1, 2\}$; in this numeration system $0$
        has only the trivial representation.
    \item $D_{\min} = \tfrac16$. It follows from the fact that, if $z_1 \neq 0$, then
        $$
        |Z| = \Bigl| \sum_{j=1}^{\infty} z_{j}\beta^{-j} \Bigr| \geq |\Im(Z)| = 2 \Bigl| \sum_{j=1}^{\infty} z_{2j-1}(-4)^{-j}\Bigr| \geq 2 \tfrac1{12} = 2 \cdot 0 \decdot \overline{1} \, \overline{2} \, 2 \, \overline{2} \, 2 \, \overline{2} \cdots \, .
        $$
\end{itemize}

Using the parameters $\varepsilon$, $D_{\min}$ and the oblong $I$ mentioned above,  the general formulas \eqref{deltatMult} and \eqref{L_calculation} for on-line multiplication give us  the delay $\delta = 9$ and the number $L = 7$ of fractional digits of $W$ to evaluate.

For on-line division, with $K = \max\{|z|: z \in I\} = \tfrac{\sqrt{146}}{9}$, using the general formulas \eqref{DeltaDivision}, \eqref{ALFA} and \eqref{L_calculation} results in the delay $\delta = 11$ and the number $L = 11$ of fractional digits of $W$ and $D$ to evaluate.

In summary, the Knuth numeration system enables on-line multiplication and division with linear time complexity. The preprocessing of divisor is just a shift of the fractional point, due to the non-existence of any non-trivial representation of zero in $(\beta, {\mathcal{A}})$. The size of the set $(\beta I)^{\varepsilon/2}$ and of the alphabet ${\mathcal{A}}$ imply that we need to evaluate another three digits of $W$ and $W/D$ before the fractional point (for on-line multiplication and division, respectively). Any point $W = \sum_{j = n}^\infty w_j \beta^{-j}$ with $w_j \in {\mathcal{A}}$ and $w_n \neq 0$ for $n \leq -3$ would lie outside the set $(\beta I)^{\varepsilon/2}$.


\subsection{Eisenstein numeration system}\label{Eisen}

The Eisenstein numeration system works with a complex base, namely $\beta = -1 + \omega $, where $\omega = \exp{\frac{2i\pi}{3}} $ is the third root of unity, i.e., $\omega^3 = 1$.

It is known that this base $\beta$ with the (so-called canonical) alphabet $\CC = \{ 0, 1, 2\}$ forms a numeration system, in which any complex number has a $(\beta, \CC)$-representation. The same property is true also for some other alphabets of cardinality $\#\CC = 3$, for example $\CC = \{ 0, 1, -\omega \}$. It follows from Theorem 3.2 in \cite{NielsenKornerup}.

Nevertheless, we choose to work with a larger, redundant (complex) alphabet ${\mathcal{A}}$ of size $\#{\mathcal{A}} = 7$:
$$
{\mathcal{A}} = \{0, \pm 1, \pm \omega, \pm \omega^2 \} \quad \ \text{with } A = \max \{ |a| \, : \, a \in {\mathcal{A}} \} = 1 \,.
$$
The numeration system $(\beta, {\mathcal{A}})$ using this alphabet has favorable properties:
\begin{itemize}
    \item the alphabet ${\mathcal{A}}$ is not just (centrally) symmetric, but also closed under multiplication;
    \item the numeration system $(\beta, {\mathcal{A}})$ enables parallel addition (and subtraction), and $\#{\mathcal{A}} = 7$ is the minimal size of alphabet allowing parallel addition for the Eisenstein base (a result to be published);
    \item there are non-trivial representations of zero in $(\beta, {\mathcal{A}})$, nevertheless, the preprocessing of divisor for on-line division is possible (as discussed in Example \ref{EisenPrepro}), and for a divisor $D$ ensures that
        $$
            \frac{\sqrt{3}(6-\sqrt{7})}{18} = D_{\min} \leq | D| \leq D_{\max} = \frac{\sqrt{7}}{2} \, .
       $$
\end{itemize}
Due to these properties, the Eisenstein numeration system with alphabet ${\mathcal{A}} = \{0, \pm 1, \pm \omega, \pm \omega^2\}$ allows on-line multiplication and division, as shown below.


\subsubsection{OL Property of Eisenstein numeration system}

For each digit $a \in {\mathcal{A}}$, we denote the set
\begin{equation*}
    H_a = \{ z \in \C : |z-a| \leq |z-b|\ \ \text{for all } b \in {\mathcal{A}}, b \neq a\} \, .
\end{equation*}

The sets $H_a$ for $a \neq 0$ are unbounded, while the set $H_0$ is the regular hexagon with center in point zero and with vertices $\pm \frac{1}{2} \pm i \frac{\sqrt{3}}{6}$ and $\pm i \frac{\sqrt{3}}{3}$.

It can be easily verified that $r = \frac{\sqrt{3}}{6}$ is the maximum possible value $r > 0$ such that
\begin{equation*}
    (\beta H_0)^{r} \subset \bigcup_{a \in {\mathcal{A}}} (H_0 + a) \, .
\end{equation*}

We work with the following $\Digit$ function:
\begin{equation*}
    \Digit(V) = a \quad \Rightarrow \quad V \in H_a \, .
\end{equation*}

\noindent Using the parameter $r = \frac{\sqrt{3}}{6}$, we can set $\varepsilon > 0$ as $\varepsilon = \frac{r}{|\beta| + 1} = \frac{3-\sqrt{3}}{12}$. Figure~\ref{Eisenstein_picture} shows that the OL Property is fulfilled with the set $I = (H_0)^\varepsilon$. Nevertheless, we modify our approach, in order to obtain optimal values for the delay $\delta$ and the number $L$ of fractional digits of arguments to evaluate in the function $\Select$.


\begin{figure}[h]\label{E}
    \centering
   \includegraphics[scale=0.45]{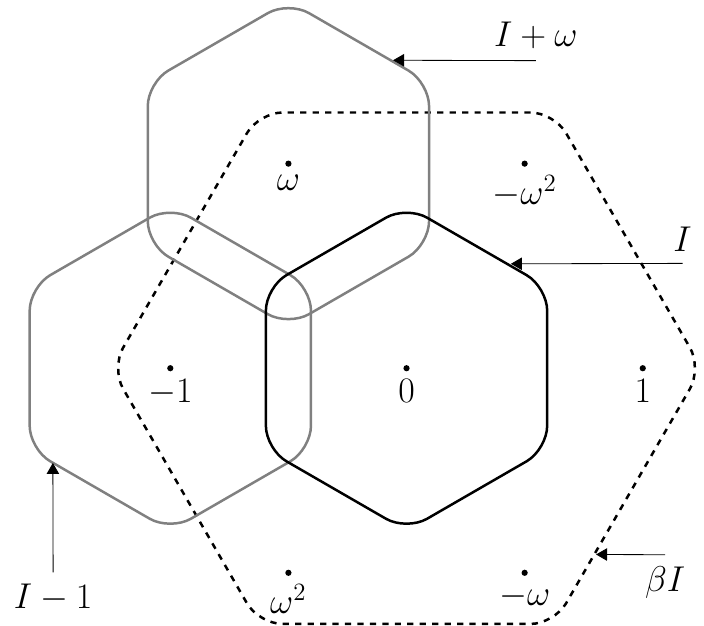}
    \caption{Eisenstein numeration system with  base $\beta = -1 + \omega$ and alphabet ${\mathcal{A}} = \{0, \pm 1, \pm \omega, \pm \omega^2\}$, where $\omega$ is the third root of unity,  fulfills the OL Property, due to the ``rounded hexagon" set $I$ illustrated hereby, see Example~\ref{Eisen}.
    \label{Eisenstein_picture}}
\end{figure}



\subsubsection{On-line multiplication in Eisenstein numeration system}

We consider two parameters $\mu, \nu > 0$ such that
\begin{equation}\label{m+n=r}
    \sqrt{3} \mu + \nu = |\beta| \mu + \nu \leq r = \frac{\sqrt{3}}{6} \, .
\end{equation}

This ensures that the set $J = (H_0)^\mu$ has the property
$$
(\beta J)^\nu = (\beta (H_0^\mu))^\nu = \beta H_0^{|\beta| \mu + \nu} \subset \beta H_0^r \subset \bigcup_{a \in {\mathcal{A}}} (H_0 + a) \, , \quad \text{and also}
$$
\begin{equation*}
    ((\beta J)^\nu)^\varphi \subset \bigcup_{a \in {\mathcal{A}}} (H_0 + a)^\varphi \quad \text{for any} \quad \varphi > 0 \, .
\end{equation*}

The selection function for multiplication $\SelectM : (\beta H_0)^r \rightarrow {\mathcal{A}}$ is defined by
$$
\SelectM(W) = \Digit(\Trunc_{\mu /2}(W)) \, ,
$$
implying that $ W - \SelectM(W) \in (H_0)^\mu $ for any $W \in (\beta H_0)^r $.

This is due to the fact that $|W - \Trunc_{\mu /2}(W)| < \frac{\mu}{2}$, and $W \in (\beta H_0)^r$ implies $V = \Trunc_{\mu /2}(W) \in (\beta H_0)^{r + \mu/2} \subset \bigcup_{a \in {\mathcal{A}}} (H_0 + a)^{\mu/2}$. Consequently, $V - \Digit(V) \in (H_0)^{\mu/2}$, and we finally obtain $W - \SelectM(W) = (W - V) + (V - \SelectM(W)) = (W-V) + (V - \Digit(V)) \in ((H_0)^{\mu/2})^{\mu/2} = (H_0)^\mu$.

In the algorithm of on-line multiplication, we perform the following iterative steps:
\begin{itemize}
    \item Put $a = \SelectM(W)$, and thus $W - a  \in (H_0)^\mu$; so $a$ is the output digit for the currently processed position;
    \item Set $W_{\new} = \beta (W - \SelectM(W)) + (xY + yX)$, where $xY + yX$ is contribution of the input operands for the next processed position; wherein
        \begin{equation*}
        W \in (\beta H_0)^r \quad \text{implies} \quad W_{\new} \in (\beta H_0)^r \, ,
        \end{equation*}
\end{itemize}
provided that $|xY+yX| \leq \nu$, and $|\beta| \mu + \nu \leq r$ (as demanded in \eqref{m+n=r}). This is readily seen, as $W - a \in (H_0)^\mu$, so $\beta(W - a) \in \beta (H_0)^\mu = (\beta H_0)^{|\beta|\mu}$, and we require $xY+yX \in (0)^\nu$; altogether
\begin{equation}\label{W_new_multiplication}
    W_{\new} = \beta (W - \SelectM(W)) + (xY+yX) \in (\beta H_0)^{|\beta|\mu} + (0)^\nu = (\beta H_0)^{|\beta|\mu + \nu} \subset (\beta H_0)^r \, .
\end{equation}

\noindent From the formulas and requirements above, we deduce the conditions determining the desired parameters $\delta$ and $L$:
\begin{itemize}
    \item From $|xY+yX| \leq \nu$, we obtain a limitation for the delay $\delta$
        \begin{equation*}
            \nu \geq \frac{2 A D_{\max}}{|\beta|^\delta} = \frac{\sqrt{7}}{\sqrt{3}^\delta} \, .
        \end{equation*}
    \item The number $L$ of fractional digits to be evaluated from the expression $W$ is limited by
        \begin{equation*}
            \mu \geq \frac{2 D_{\max}}{|\beta|^L} = \frac{\sqrt{7}}{\sqrt{3}^L} \, .
        \end{equation*}
\end{itemize}
At the same time, we have to maintain the inequality $\sqrt{3}\mu + \nu = |\beta|\mu + \nu \leq r = \frac{\sqrt{3}}{6}$ --- so the bigger part of $r$ we dedicate to $\delta$ via $\nu$, the lesser part remains for $L$ via $\mu$. Depending on this distribution, we find two reasonable combinations of the parameters $L$ and $\delta$ in the algorithm of on-line multiplication in Eisenstein numeration system:
\begin{itemize}
    \item $(\delta_{\min}, L) = (5, 7)$, where in the delay $\delta$ is minimized; and
    \item $(\delta, L_{\min}) = (6, 6)$, where the parameter $L$ is minimized.
\end{itemize}


\subsubsection{On-line division in Eisenstein numeration system}

When specifying the algorithm for on-line division, we use again the general formula \eqref{DeltaDivision}. The $\Trunc$ function provides partial evaluations $V = \Trunc_\alpha(W)$ and $\Delta = \Trunc_\alpha(D)$, where the parameter $\alpha > 0$ is set so that $\left| \frac{W}{D} - \frac{V}{\Delta} \right| \leq \frac{\mu}{2}$. We set another auxiliary parameter
\begin{equation}
    K = \max \{ |z| \, : \, z \in H_0 \} = \frac{\sqrt{3}}{3}\, .
\end{equation}

During the course of the iterations of the algorithm, it shows that
\begin{equation}
    W \in D (\beta H_0)^\nu \quad \text{ implies } \quad W_{new} \in D_{new} (\beta H_0)^\nu \, ,
\end{equation}
provided that $\mu, \nu >0$ fulfill~\eqref{m+n=r}. The inequalities translating relations between parameters $\mu, \nu$ and the desired results $\delta$ and $L$ are somewhat more laborious here than in the case of on-line multiplication:
\begin{equation}
    \nu \geq \frac{A(D_{max} + 1 + K + \mu)}{D_{min} |\beta|^\delta} \quad \text {and} \quad \mu \geq \frac{2 D_{max} (|\beta| K + r+ 1)}{D_{min} |\beta|^L} \, .
\end{equation}

Depending on distribution of the value $r$ between $\mu$ and $\nu$, according to~\eqref{m+n=r}, we obtain two reasonable combinations of the parameters $L$ and $\delta$ in the on-line division algorithm for the Eisenstein numeration system:
\begin{itemize}
    \item $(\delta_{\min}, L) =  (7, 10)$, where the delay $\delta$ is minimized; and
    \item $(\delta, L_{\min}) = (10, 9)$, where the parameter $L$ is minimized.
\end{itemize}


\section{Conclusion}\label{conclusion}

It is known that many continuous functions of  real variables can be calculated by an on-line algorithm in a redundant numeration system. For a  precise definition of redundancy of a numeration system, formalization of on-line computation and results, see \cite[Chapter 2]{Kurka}. In particular, multiplication and division are on-line computable. However, this general result does not provide any effective algorithm for calculation. The exceptionality of the algorithms due to Trivedi and Ercegovac consists in their linear time complexity, i.e., the number of steps needed to compute the first $n$ most significant digits of the result is $\mathcal{O}(n)$.

These algorithms were originally introduced for numeration systems $(\beta, A)$ where $\beta$ is a natural integer. We have shown that they can be extended to real or complex systems as well, provided that $(\beta, {\mathcal{A}})$ has the OL Property. Investigating the OL Property and defining the preprocessing rules for a given system $(\beta, {\mathcal{A}})$ remains an open problem, particularly if we want to use a digit set ${\mathcal{A}}$ minimal in size. On several examples we have demonstrated that the existence of convenient preprocessing rules, together with parallel addition and subtraction, implies linear time complexity for both algorithms. Recently  in \cite{FrPe}, we described bases for which  symmetric alphabets of consecutive  integers allow preprocessing.   Nevertheless, identifying the numeration systems for which the algorithms of Trivedi-Ercegovac work in linear time, need a deeper study.


\section*{Acknowledgments}

The first  author acknowledges support of ANR/FWF project ``FAN",  E.P.   financial support of project no. CZ.02.1.01/0.0/0.0/16\_019/0000778, E.P. and M.S.  support of GA\v CR 13-03538S,   M.P. support of SGS 11/162/OHK4/3T/14.



\end{document}